\documentclass[11pt,a4paper]{article}

\usepackage[truedimen,margin=32mm]{geometry} 

\usepackage{mathrsfs}
\usepackage{amssymb}
\usepackage{amsmath}
\usepackage{ascmac}
\usepackage{amsthm}
\usepackage[pdftex]{graphicx}
\usepackage{natbib}
\usepackage{color}
\usepackage{setspace}
\usepackage{url}
\usepackage{here}
\usepackage{comment,booktabs,bm}
\usepackage{lscape}

\usepackage{titlesec}
\titleformat*{\section}{\large\bfseries}
\titleformat*{\subsection}{\it}

\newtheorem{thm}{Theorem}
\newtheorem{lem}{Lemma}

\newtheorem{prp}{Proposition}

\newtheorem{rem}{Remark}

\def\At{{\widetilde A}}
\def\ta{{\tau}}

\def\pd{\partial}

\def\al{{\alpha}}
\def\be{{\beta}}
\def\ga{{\gamma}}
\def\de{{\delta}}
\def\ep{{\varepsilon}}
\def\la{{\lambda}}

\def\fai{{\varphi}}
\def\ka{{\kappa}}

\def\kah{{\hat \ka}}

\def\lat{{\tilde \la}}

\def\Ga{{\Gamma}}

\def\non{{\nonumber}}
\title{{\bf Sparse Bayesian inference on gamma-distributed observations using shape-scale inverse-gamma mixtures}}
\author{}
\date{}

\begin{document}

\maketitle
\doublespacing

\vspace{-1.5cm}
\begin{center}
{\large Yasuyuki Hamura$^{1}$, Takahiro Onizuka$^{2}$,\\ Shintaro Hashimoto$^{2}$ and Shonosuke Sugasawa$^{3}$}
\end{center}

\medskip
\noindent
$^{1}$Graduate School of Economics, Kyoto University\\
$^{2}$Department of Mathematics, Hiroshima University\\
$^{3}$Center for Spatial Information Science, The University of Tokyo

\medskip
\medskip
\medskip
\begin{center}
{\bf \large Abstract}
\end{center}

In various applications, we deal with high-dimensional positive-valued data that often exhibits sparsity.  
This paper develops a new class of continuous global-local shrinkage priors tailored to analyzing gamma-distributed observations where most of the underlying means are concentrated around a certain value. 
Unlike existing shrinkage priors, our new prior is a shape-scale mixture of inverse-gamma distributions, which has a desirable interpretation of the form of posterior mean and admits flexible shrinkage.  
We show that the proposed prior has two desirable theoretical properties; Kullback-Leibler super-efficiency under sparsity and robust shrinkage rules for large observations.
We propose an efficient sampling algorithm for posterior inference. 
The performance of the proposed method is illustrated through simulation and two real data examples, the average length of hospital stay for COVID-19 in South Korea and adaptive variance estimation of gene expression data.


\newpage

\section{Introduction}

\label{sec:introduction}
In various statistical applications, we often face a sequence of positive-valued observations such as machine failure time, store waiting time, survival time under a certain disease, an income of a certain group, and so on. 
A common feature of the data is ``sparsity" in the sense that most of the underlying means of observations are concentrated around a certain value (grand mean) while a small part of the means is significantly away from the grand mean. 
To reflect the sparsity structure, a useful Bayesian technique is an idea of ``global-local shrinkage" \cite[e.g.][]{polson2012local} that provides adaptive and flexible shrinkage estimation of underlying means; when the observations are around the grand mean, the posterior mean strongly shrinks the observation toward the grand mean, but the observations that are away from the grand mean remain unshrunk.  

This paper proposes a new framework for sparse Bayesian inference on a sequence of positive-valued observations by using gamma sampling distributions for observations and develops a novel class of global-local shrinkage priors for positive-valued heterogeneous mean parameters based on shape-scale mixtures of inverse-gamma distributions. 
Specifically, we introduce a scaled beta (SB) distribution and its extension called inverse rescaled beta (IRB) distribution as mixing distributions in the shape-scale mixture. 
We discuss distributional properties, tail decay rate, and concentration around the origin of the proposed priors and develop an efficient sampling scheme from the posterior distribution. 
Moreover, we reveal two theoretical properties of the proposed prior, tail-robustness for large means and Kullback-Leibler supper-efficiency under sparsity.

There are several works on shrinkage inference of a sequence of positive-valued data. 
Under the gamma sampling model (as in our proposal), simultaneous estimation for rate/scale parameters was considered by several authors decades ago \citep[e.g.,][]{Berger1980, Ghosh1980, DasGupta1986, Dey1986}.
However, the classical framework does not take into account sparsity and provides only universal shrinkage regardless of the observed values. 
To address the sparsity in positive-valued data, \cite{Donoho2006} proposed a threshold-type estimator with the false discovery rate control, but the sampling model is an exponential distribution (a special case of gamma distribution).
Therefore, its applicability is quite limited. 
More recently, \cite{Lu2016} proposed an empirical Bayes shrinkage method customized for variance estimation using a $\chi^2$-distribution (a special case of gamma distribution) for the observed sampling variance and a finite mixture of inverse-gamma priors for the true variance. 
However, this approach does not address sparsity, and no theoretical results are discussed. 
More importantly, the existing methods only produce point estimates of the underlying means. In contrast, the proposed method can obtain full information on posterior distributions, enabling us to carry out uncertainty quantification.

In Bayesian analysis, the methodology and application of ``global-local shrinkage priors" have been developed last decades. 
Under Gaussian sequence or normal linear regression models, there have been a variety of shrinkage priors including the most famous horseshoe \citep{Carvalho2010} prior and its related priors \citep[e.g.][]{Armagan2013,bhadra2017horseshoe+,bhattacharya2015dirichlet,hamura2020shrinkage,zhang2020bayesian}.
Such prior is known to have an attractive shrinkage property, making it possible to strongly shrink small observations toward zero while keeping large observations unshrunk. 
Recently, techniques of global-local shrinkage priors for Gaussian data are extended to the (quasi-)sparse count data \citep[e.g.][]{Datta2016, Hamura2019}. 
Although several theoretical properties (e.g., Kullback-Leibler supper-efficiency and tail-robustness) have been revealed under the Gaussian and Poisson sampling distributions, theoretical properties of global-local shrinkage under the gamma sampling model are not fully discussed.
Furthermore, the theoretical development of the proposed prior requires substantial work due to the form of shape-scale mixtures that are rather different from the existing global-local shrinkage priors. 
To fill the gap, we contribute to the theoretical development of global-local shrinkage by showing Kullback-Leibler supper-efficiency and tail-robustness under the gamma sampling model.

The remainder of the paper is structured as follows. 
In Section \ref{sec:model}, we introduce settings and our hierarchical model, and we propose a global-local shrinkage prior based on a kind of beta distribution. Furthermore, we illustrate the properties of the marginal prior and posterior distributions for $\lambda_i$, and also discuss the selection of hyperparameters of the proposed priors. An efficient posterior computation algorithm is constructed via the Markov chain Monte Carlo method. In Section \ref{sec:theory}, we show two theoretical properties of the proposed priors. 
The performance of the proposed method is demonstrated through numerical studies in Section \ref{sec:sim}, and we apply the method to two real datasets related to the average length of hospital stay for COVID-19 in South Korea and variance estimation of gene expression data in Section \ref{sec:data}. 
Proofs and technical details are given in the Supplementary Material. 
R code implementing the proposed methods is available at Github repository (\url{https://github.com/sshonosuke/GLSP-gamma/}).

\section{Sparse Bayesian inference on gamma-distributed observations }
\label{sec:model} 

\subsection{Settings and models}
\label{subsec:model} 
Suppose we observe a sequence of gamma-distributed observations, denoted by $y_1,\ldots,y_n$.
For each $i = 1, \dots , n$, we assume the following gamma model $y_i$:
\begin{equation}\label{gamma}
y_i\mid\lambda_i \sim {\rm{Ga}} \left(\delta_i , \frac{\delta_i}{\lambda_i\eta_i}\right) , 
\end{equation}
where ${\rm Ga}(\alpha, \beta)$ denotes a gamma distribution with shape parameter $\alpha$ and rate parameter $\beta$, $\delta_i$ is a fixed constant, and $\lambda_i$ is a parameter of interest.
Under the model, $E(y_i)=\lambda_i\eta_i$ and $\eta_i$ is a structural component that may be modeled to incorporate covariates and other external information (e.g., spatial information). 
In what follows, we assume $\eta_i=1$ for simplicity, under which $\lambda_i$ is interpreted as the mean of $y_i$, but all the computation algorithms and analytical results are valid for the general form of $\eta_i$ as long as $\eta_i$ is conditioned on. 
As considered in \cite{Lu2016}, if $y_i$ and $\lambda_i$ are sampling and true variances, respectively, the choice is $\delta_i=n_i/2$, where $n_i$ is a sample size used to compute $y_i$.
Moreover, if $y_i$ is a sample mean based on $n_i$ samples generated from an exponential distribution ${\rm Exp}(1/\lambda_i)$, it holds that $\delta_i=n_i$, and it reduces the framework of a sequence of exponential data when $n_i=1$, considered in \cite{Donoho2006}.  
In the present framework, our interest lies in the simultaneous estimation of the sequence of positive-valued means $\lambda=(\lambda_1,\ldots,\lambda_n)$ by combining information of a given set of data $y = ( y_1 , \dots , y_n )$. 
In particular, we focus on the structure that most observations are located around the grand mean while some observations are very large. 
To carry out flexible Bayesian inference even under this situation, we employ an idea of global-local shrinkage that can provide customized shrinkage estimation of $\lambda_i$ depending on the location of the observed value $y_i$.

Specifically, we consider the following prior distribution for $\lambda_i$: 
\begin{equation}\label{prior-lam}
\lambda_i\mid u_i\sim {\rm{IG}} ( 1 + \tau u_i , \beta \tau u_i ), \ \ \ \ i=1,\ldots,n,
\end{equation}
where $\beta$ and $\tau$ are unknown global parameters and $u_i$ is a local parameter related to the customized shrinkage rule. 
The prior mean of $\lambda_i$ is $E(\lambda_i)=\beta$ so that $\beta$ is interpreted as a grand mean of underlying heterogeneous means. 
On the other hand, since ${\rm Var}(\lambda_i)=\beta^2/(\tau u_i-1)$ as long as $\tau u_i>1$, $\tau$ and $u_i$ control the scale of the prior. 
Unusual parametrization of (\ref{prior-lam}) is the dependence of both shape and scale parameters on the local parameter $u_i$ so that setting a mixing distribution for $u_i$ leads to a class of shape-scale mixtures of inverse-gamma distributions.  
However, this parametrization is essential to interpret the form of posterior means of $\lambda_i$.

Under the inverse-gamma prior (\ref{prior-lam}), the conditional posterior distributing of $\lambda_i$ given $u_i$ is ${\rm IG}(1+\delta_i+\tau u_i, \delta_iy_i+\beta \tau u_i)$, so that the posterior mean of $\lambda_i$ is given by 
\begin{align*}
E(\lambda_i \mid y_i )  
= E\left(\frac{\delta_i y_i+\beta \tau u_i}{\delta_i+\tau u_i} \mid y_i\right)
= \beta + \left\{1- E(\kappa_i\mid y_i)\right\}(y_i-\beta),
\end{align*}
where $\kappa_i=\tau u_i/(\delta_i + \tau u_i)\in (0,1)$ is known as {\it shrinkage factor} that determines the amount of shrinkage of $y_i$ toward the grand mean $\beta$. 
As desirable properties of $\kappa_i$, $E(\kappa_i\mid y_i)$ should be close to $1$ when $y_i$ is close to the grand mean, leading to strong shrinkage toward $\beta$, while $E(\kappa_i\mid y_i)$ should be sufficiently small for $y_i$ having large $y_i-\beta$ to prevent bias caused by over-shrinkage.  
We also note that the global parameter $\tau$ determines the overall shrinkage effect, whereas the local parameter $u_i$ allows $\kappa_i$ to vary over different observations.

\subsection{Global-local shrinkage priors}
\label{subsec:prior}

Our hierarchical model can be expressed as
\begin{align*}
y_i\mid\lambda_i \sim {\rm{Ga}} \left(\delta_i , \frac{\delta_i}{\lambda_i}\right),\ \ \ \
\lambda_i\mid u_i\sim {\rm{IG}} ( 1 + \tau u_i , \beta \tau u_i ), \ \ \ \
u_i\sim \pi(\cdot), 
\end{align*}
where priors for $\beta$ and $\tau$ are discussed at the end of this subsection.
For the local parameter $u_i$, we suggest two prior distributions. 
The first one is the scaled beta (SB) prior 
\begin{align*}
\pi_{\mathrm{SB}} ( u_i ) = \frac{1}{B(a, b)} \frac{{u_i}^{a - 1}}{(1 + u_i )^{a + b}},
\end{align*}
where $a,b>0$ are hyperparameters and $B(a,b)$ is the beta function. 
The SB distribution is also known as the beta prime distribution \citep[e.g.][]{Johnson1995}, and the related family of distributions has been often used in Bayesian statistics \citep[e.g.][]{Perez2017, Hamura2021}, especially in the context of shrinkage priors. 
As an alternative prior, we newly propose the inverse rescaled beta (IRB) prior
\begin{align*}
\pi_{\mathrm{IRB}} ( u_i ) &= \frac{1}{B(b, a)} \frac{1}{u_i (1 + u_i )} \frac{\{ \log (1 + 1 / u_i ) \} ^{b - 1}}{\{ 1 + \log (1 + 1 / u_i ) \} ^{b + a}}.
\end{align*}
Note that the IRB prior for $u_i$ is equivalent to using the rescaled beta prior \citep{Hamura2021} for $1 / u_i$.

Here, we summarize basic properties of $\pi_{\mathrm{SB}} ( u_i )$ and $\pi_{\mathrm{IRB}} ( u_i )$ under $u_i\to 0$ and $u_i\to\infty$. 
As is well known, the SB prior has the following properties.
\begin{itemize}
\item
\textit{Concentration at the origin}. 
As $u_i \to 0$, we have $\pi_{\mathrm{SB}} ( u_i ) \propto {u_i}^{a - 1}$. 
In particular, $\pi_{\mathrm{SB}} ( \kappa_i ) \to \infty $ as $\kappa_i \to 0$ if and only if $a < 1$. 
\item
\textit{Tail decay}. 
As $u_i \to \infty $, we have $\pi_{\mathrm{SB}} ( u_i ) \propto {u_i}^{- 1 - b}$. 
In particular, $\pi_{\mathrm{SB}} ( \kappa_i ) \to \infty $ as $\kappa_i \to 1$ if and only if $b < 1$. 
\end{itemize}
Meanwhile, ignoring log factors, we see that $\pi_{\mathrm{IRB}} ( u_i )$ has the following properties: 
\begin{itemize}
\item
\textit{Concentration at the origin}. 
As $u_i \to 0$, we have  
$\pi_{\mathrm{IRB}} ( u_i ) \approx {u_i}^{- 1}$. 
In particular, 
$\pi_{\mathrm{IRB}} ( \kappa_i ) \approx {\kappa_i}^{- 1} \to \infty $ as $\kappa_i \to 0$ whatever the value of $a > 0$ is. 
This is in contrast to the case of the SB prior. 
\item
\textit{Tail decay}. 
As $u_i \to \infty $, we have $\pi_{\mathrm{IRB}} ( u_i ) \propto {u_i}^{-1 - b}$. 
In particular, $\pi_{\mathrm{IRB}} ( \kappa_i ) \to \infty $ as $\kappa_i \to 1$ if and only if $b < 1$. 
This is exactly as in the case of the SB prior. 
\end{itemize}

In the context of existing global-local shrinkage priors, the concentration at both $\kappa_i=0$ and $\kappa_i=1$ is closely related to the properties of shrinkage and tail robustness of the marginal prior of the parameter of interest \cite[e.g.][]{Carvalho2010,Datta2016}. 
However, as shown in the subsequent section, the concentration at $\kappa_i=0$ leads to unnecessary shrinkage toward the origin in our framework, possibly because the local parameter depends not only on scale but also on shape unlike the existing formulation of global-local shrinkage. 
Hence, we should not pursue the concentration at $\kappa_i=0$ in the proposed model.
In fact, as shown in Proposition~\ref{prp:marginal}, the choice of $a$ (controlling the concentration at $\kappa_i=0$) is not related to the performance of shrinkage and tail robustness as the marginal prior of $\lambda_i$.

We discuss the priors for $\beta$ and $\tau$.
Remember that $\beta$ is a grand mean (i.e. shrinkage target of the posterior mean) of $\lambda_i$ and $\tau$ controls the overall shrinkage. 
It would be possible to fix $\beta$ or assign an informative prior for $\beta$ if the user has much information about $\beta$.
On the other hand, when there is not much prior information on $\beta$ and $\tau$, we recommend using proper but slightly diffuse priors.
In our numerical studies, we use priors, $\beta\sim {\rm Ga}(0.1, 0.1)$ and $\tau\sim {\rm Ga}(0.1, 0.1)$ as default priors, which are conditionally conjugate. 
Although improper priors can be assigned for $\beta$ and $\tau$, checking the posterior propriety given a certain form of improper prior is not straightforward due to the complicated hierarchical forms of the model.  
In the Supplementary Material, we discuss the conditions of posterior propriety under some forms of improper priors.
For example, using $\pi ( \beta ) \propto 1/\beta$ combined with a proper gamma prior for $\tau$ leads to posterior propriety under some conditions. 
Furthermore, we also note that the standard improper priors for scale parameters such as $\pi ( \beta ) \propto 1$ or $\pi ( \beta ) \propto 1 / \beta $ may not be necessarily reasonable under the hierarchical gamma model, that is, it is not clear whether these priors can be justified as objective ones such as reference priors.
Since we assume subjective priors for $\lambda_i$ and $u_i$, we may be able to consider reference priors for $\beta$ and $\tau$ using an idea of partial information prior \citep{sun1998reference}, but we do not pursue the detailed argument here.

\subsection{Marginal prior for $\lambda_i$}
\label{subsec:marginal_prior} 

In this section, we consider the behavior of the marginal prior of $\lambda_i$.
We assume that the grand mean $\beta$ and global shrinkage parameter $\tau$ are fixed at $1$ for simplicity so that the grand mean is $1$ in the following discussion. 
We first discuss the roles of the hyperparameters, $a$ and $b$, of the proposed priors, and then we propose particular choices of the hyperparameters.

The goal is to select $a$ and $b$ so that the marginal prior of $\lambda_i$ should ideally (G1) not be thick at the origin and have (G2) a fat right-tail and (G3) a spike at $1$. We provide the following analytical results concerning the behavior of the marginal prior for $\lambda_i$.

\begin{prp}
\label{prp:marginal} 
Suppose that either $\pi(u_i)=\pi_{\rm{SB}}(u_i) \propto {u_i}^{a - 1} / (1 + u_i )^{a + b}$ or $\pi(u_i)=\pi_{\rm{IRB}}(u_i) \propto [1 / \{ u_i (1 + u_i ) \} ] \{ \log (1 + 1 / u_i ) \} ^{b - 1} / \{ 1 + \log (1 + 1 / u_i ) \} ^{b + a}$. 
Then the marginal prior $p(\lambda_i)$ of $\lambda_i$ has the following properties:

\medskip
\noindent
{\rm (i)} As $\lambda_i \to 0$, 
\begin{align*}
p(\lambda_i) \approx \begin{cases} {\lambda_i}^{a - 1}, & \text{if \ $\pi(u_i) = \pi _{\rm{SB}} ( u_i )$}, \\ {\lambda_i}^{- 1}, & \text{if \ $\pi ( u_i ) = \pi _{\rm{IRB}} ( u_i )$}. \end{cases}
\end{align*}
\noindent
{\rm (ii)} As $\lambda_i \to \infty $, 
\begin{align*}
p(\lambda_i) \approx {\lambda_i}^{- 2}.
\end{align*}
\noindent
{\rm (iii)} As $\lambda_i \to 1$, 
\begin{align*}
p(\lambda_i) \to \begin{cases} \infty \text{,} & \text{if \ $b \le 1 / 2$}, \\ C_1 < \infty \text{,} & \text{if \ $b > 1 / 2$} \end{cases}
\end{align*}
for some finite positive constant $0 < C_1 < \infty $. 
\end{prp}

Proposition \ref{prp:marginal} shows that our three goals are achieved whenever $b \le 1 / 2$ and $a > 1$ for the SB prior, and that goal (G1) is impossible to achieve under the IRB prior but (G2) and (G3) are achieved for $b \le 1 / 2$. 
This result is obtained from a more general theorem (Theorem S1), given in the Supplementary Material. 
Theorem S1 provides equivalents for the tail densities and density at $1$ of $\lambda $ under different priors for $u_i$ and relies on convergence theorems and approximations to prove them. 
We note that log factors are ignored in the above statement.

In more detail, Part (i) corresponds to shrinkage and non-shrinkage for small $\lambda_i$ under the SB prior with $a > 1$ and the IRB prior, respectively. 
Part (ii) corresponds to robustness for large $\lambda_i$ (i.e., the posterior mean of $\lambda_i$ does not shrink large $y_i$) under the proposed priors; if we fix $u_i$, then we necessarily have $p( \lambda_i ) \propto {\lambda_i}^{- 2 - u_i} < {\lambda_i}^{- 2}$ as $\lambda_i \to \infty $. 
Part (iii) corresponds to shrinkage for moderate $\lambda_i$ under the proposed priors with $b \le 1 / 2$; if we fix $u_i$, then $p( \lambda_i )$ never diverges at $\lambda_i = 1$. 

In other words, the left tail of $\pi(u_i)$ can affect the left tail of $p(\lambda_i)$ if we use the SB prior with $a \le 1$ or the IRB prior; the right tail of $p(\lambda_i)$ is guaranteed to be sufficiently heavy for any values of the hyperparameters; we can expect that a sufficient amount of prior probability mass is put around $\lambda_i = 1$ if we choose $b \le 1 / 2$ for the SB and IRB priors. 
Based on these findings, we propose to use $a > 1$ for the SB prior and $b \le 1 / 2$ for both the SB and IRB priors. 
In particular, our default choices are $a = 2$ and $b = 1 / 2$ for both the priors.

The marginal prior densities of $\lambda_i$ under the SB and IRB priors are illustrated in Figure \ref{fig:prior}.  
As expected, it can be seen from the right panel that the right tail of $p(\lambda_i)$ is heavier under the proposed priors than the global shrinkage prior (denoted by GL in Figure \ref{fig:prior}) when $u_i$ is fixed, that is, $\lambda_i \sim {\rm{IG}} (2, 1)$. 
Also, it is confirmed that the IRB prior makes the right tail heavier than the SB prior. 
The left panel shows that the hyperparameter $a$ of the SB and IRB priors causes a trade-off between undesirable tail thickness at the origin and desirable tail thickness at infinity. 
However, for the case of the SB prior, we at least have that $p(\lambda_i) \to 0$ as $\lambda_i \to 0$ for $a = 2$ and for $a = 3$. 
The most remarkable point we want to stress here is that under each of the proposed priors, %
$p(\lambda_i)$ has a spike at $\lambda_i = 1$. 
This means that a large shrinkage effect is expected when we use one of the proposed priors, and this is quite in contrast to the case of fixing $u_i = 1$, where the mode of $p( \lambda _i )$ is significantly shifted to the left. 

Finally, the choice $a = 2$ may seem slightly strange in the literature on global-local shrinkage priors. 
Under $u_i\sim \mathrm{SB}(a,b)$, the shrinkage factor $\kappa_i=u_i/(1+u_i)$ follows the beta distribution $\mathrm{Beta}(a,b)$. 
The well-known horseshoe prior \citep{Carvalho2010} corresponds to the case $(a,b)=(1/2,1/2)$, and the resulting prior distribution of $\kappa_i$ is $\mathrm{Beta}(1/2,1/2)$, which has the popular U-shaped density. 
For our model, we do not adopt the choice $(a,b)=(1/2,1/2)$, since 
setting $a = 1 / 2$ causes 
unexpected tail-robustness (or lack of desirable shrinkage toward the grand mean) around the origin and since using $a > 1$ does not affect tail-robustness around infinity much (see also Section \ref{subsec:tail-robustness}). 

\begin{figure}[!htb]
\centering
\includegraphics[width = \linewidth]{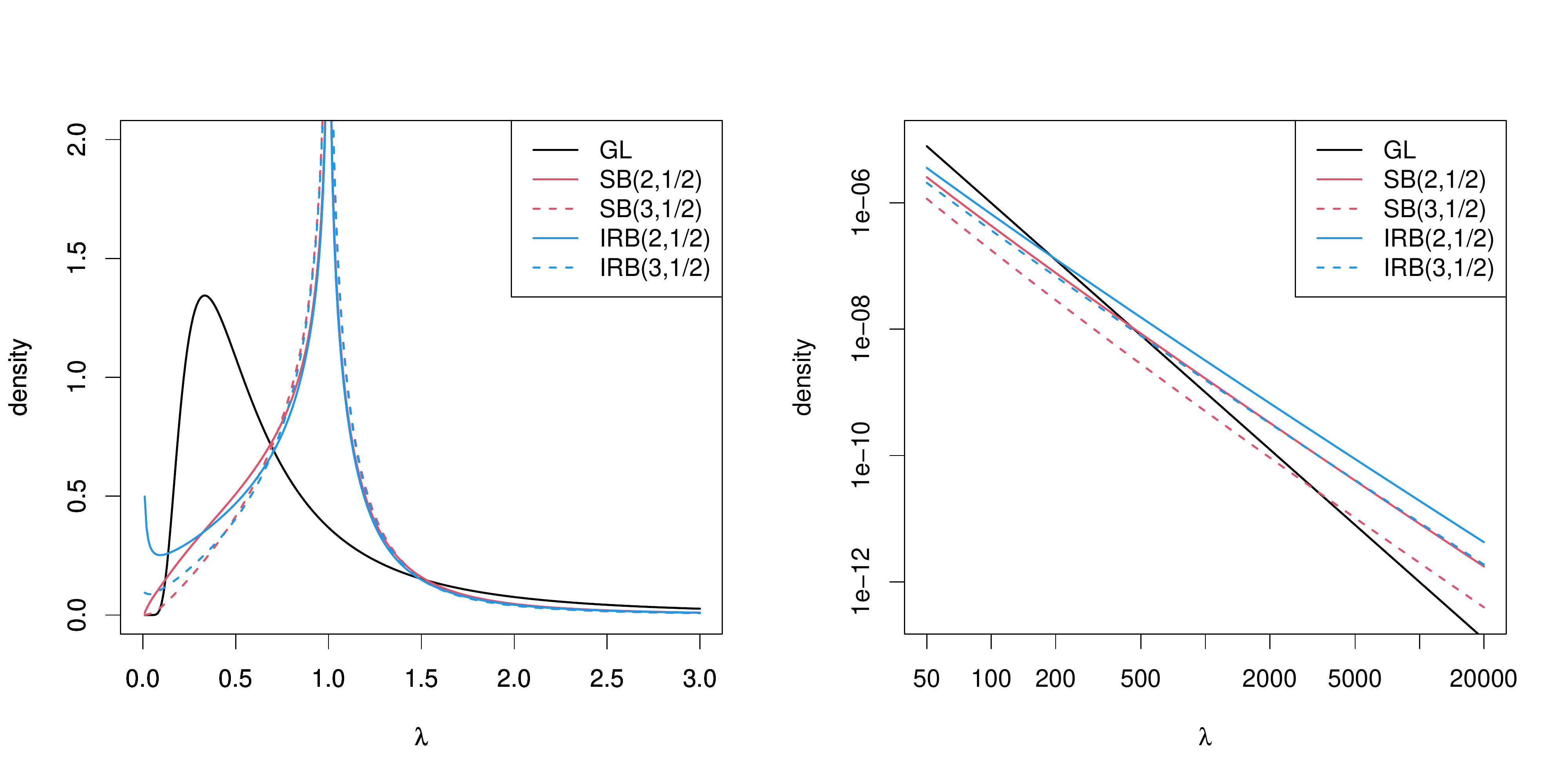}
\caption{Marginal prior densities for $\lambda_i$. The right panel is an enlarged version of the left panel in the log-scaled tail region.}
\label{fig:prior}
\end{figure}

\subsection{Marginal posterior of $\lambda_i$}
\label{subsec:marginal_posterior} 
Here, we discuss the flexibility of the proposed prior distributions. 
As an artificial example, we suppose that $m=50$ observations (the first 46 observations are $5$ and the others are $7, 15, 30, 50$) are observed. 
Furthermore, we set $\delta_i = 5$.
We show marginal posterior distributions of the shrinkage factor $\kappa_i$ given $y \in \{ 7, 15, 30, 50\} $ in Figure \ref{fig:posterior}. 
The marginal posterior under the global shrinkage prior ($u_i=1$) does not depend on $y$ and over-shrinks the posterior density under a large signal such as $y=30$ and $y=50$. Also, the global shrinkage method does not have strong shrinkage near the grand mean when $y=7$.
On the other hand, Figure \ref{fig:posterior} shows that the posterior of $\kappa_i$ under the SB and IRB priors change flexibly according to the observed values, as expected from the design of the priors. 
Comparing the two priors, it can be seen that the IRB posterior is more concentrated around $\kappa_i=0$ than the SB prior when $y_i$ is large.
Therefore, we recommend using IRB prior to situations where tail-robustness is required.

We further investigate the behavior of the posterior distribution through posterior means and variances of $\lambda_i$ as a function of $y_i$.
To see the properties of the local shrinkage property, the hyperparameters in the three priors are fixed to their posterior means obtained to make Figure \ref{fig:posterior}.
We set $(\log y_1,\ldots,\log y_{100})$ to equally-spaced 100 points from $-4$ to $4$, and computed posterior means and variances of $\lambda_i  \ (i=1,\ldots,100)$ based on the three priors. 
The results are shown in Figure \ref{fig:pos-mean}.
It is observed that both the posterior mean and variance of the GL prior are simple functions of $y_i$.
On the other hand, the proposed two priors, SB and IRB, strongly shrink $y_i$ around the grand mean while do not shrink large or small $y_i$. 
Moreover, the posterior variances of the proposed two priors are small around the grand mean due to the strong shrinkage property and those are large when the observed value is large.

\begin{figure}[!htb]
\centering
\includegraphics[width = \linewidth]{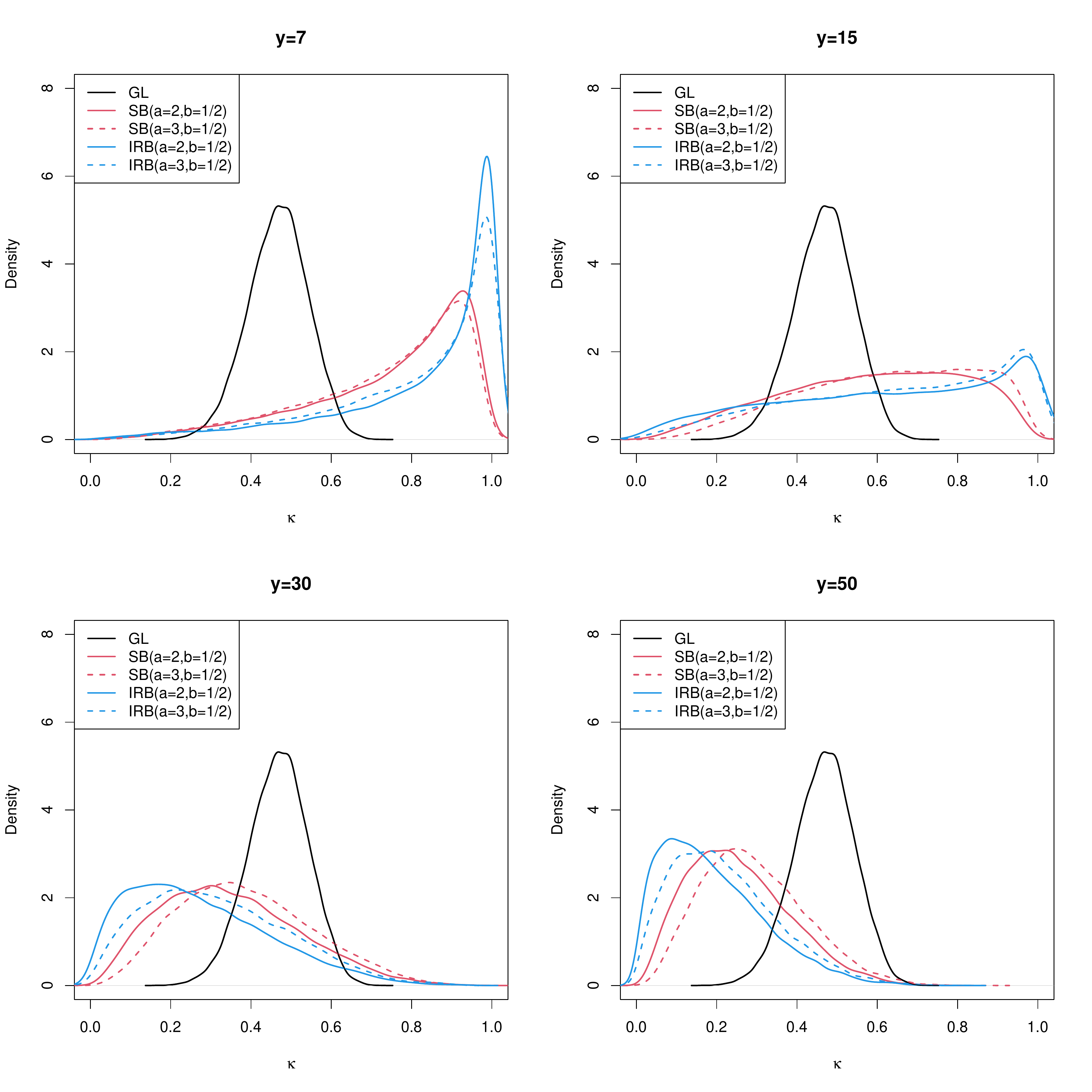}
\caption{Marginal posterior densities for the shrinkage factor $\kappa_i$ under four types of observed values.}
\label{fig:posterior}
\end{figure}

\begin{figure}[!htb]
\centering
\includegraphics[width = \linewidth]{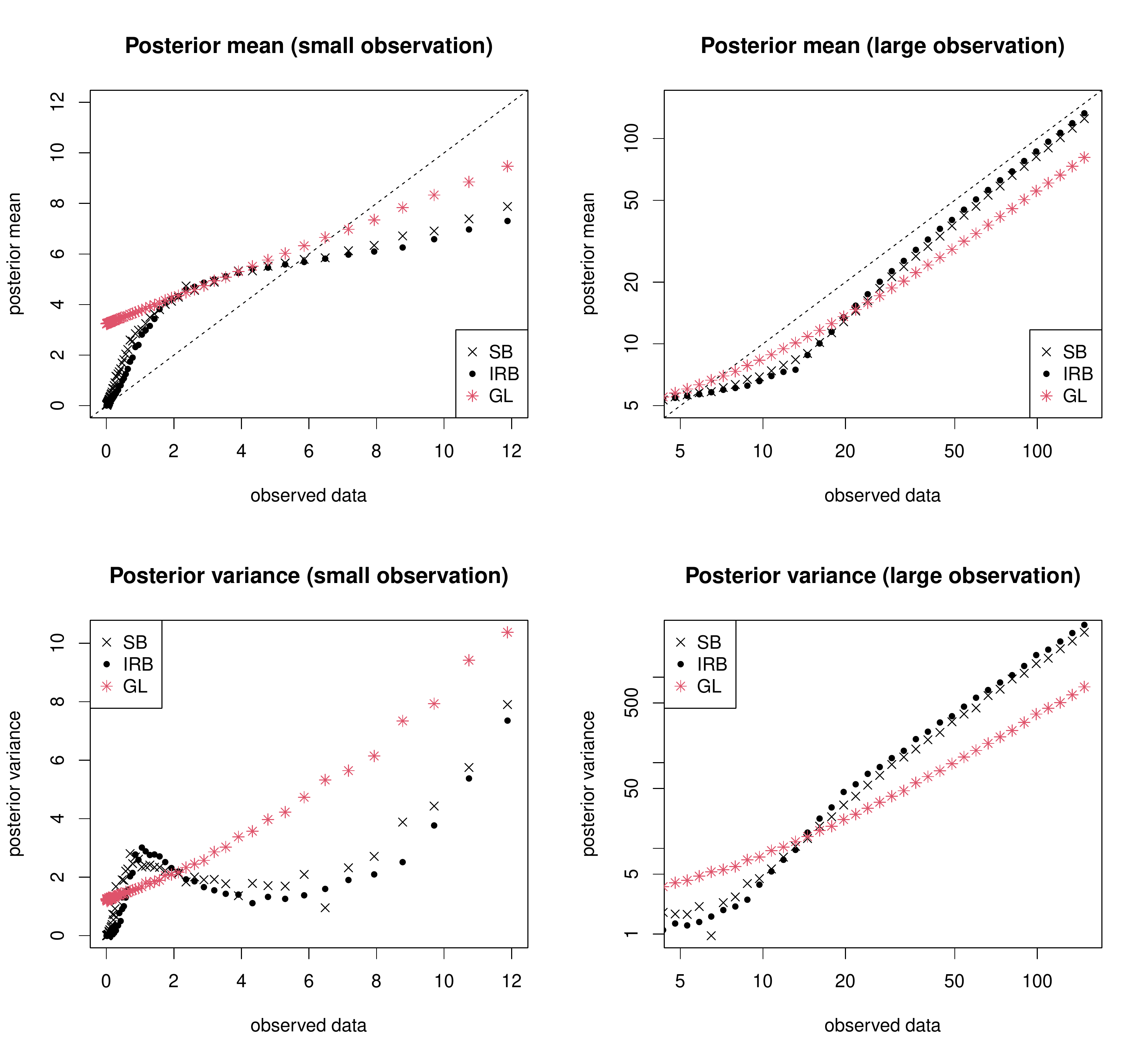}
\caption{Posterior means and variances of $\lambda_i$ for various observed values. }
\label{fig:pos-mean}
\end{figure}

\subsection{Posterior computation}
\label{subsec:mcmc} 
We provide an efficient Metropolis within the Gibbs algorithm for our model by using the approximation method of \cite{Miller2019}. 
Here, we consider the case of the SB prior. 
The details of posterior computation under the IRB prior are given in the Supplementary Material. 
In order to simplify sampling of $\tau$, we make the change of variables $\nu _i = \tau u_i$ for $i = 1, \dots , n$. 
Then the overall posterior distribution of $(\lambda,\beta,\tau, \nu )$ given $y $ is expressed by
\begin{align*}
p( \lambda , \beta , \tau , \nu \mid y ) &\propto \pi( \beta ) \pi( \tau ) \frac{1}{\tau ^n} \prod_{i = 1}^{n} \left\{ \pi ( \nu _i / \tau ) \frac{\beta ^{\nu _i + 1} {\nu _i}^{\nu _i}}{\Gamma ( \nu _i )} \frac{1}{{\lambda_i}^{\nu _i + 2}} e^{- \beta \nu _i / \lambda_i} \frac{1}{{\lambda_i}^{\delta_i}} \exp \left( - \frac{\delta_i y_i}{\lambda_i} \right) \right\} , 
\end{align*}
where $\nu=(\nu_1,\dots, \nu_n)$. 
Since the SB prior density is expressed as
\begin{align*}
\pi_{\mathrm{SB}} ( u_i ) &= \frac{1}{B(a, b)} \frac{{u_i}^{a - 1}}{(1 + u_i )^{a + b}} = \frac{1}{\Gamma(a)\Gamma(b)} \int_{0}^{\infty } {t_i}^{a + b - 1} e^{- t_i} {u_i}^{a - 1} e^{- t_i u_i} \mathrm{d}{t_i}
\end{align*}
for all $i = 1, \dots , n$, it follows that 
\begin{align*}
p( \lambda , \beta , \tau , \nu \mid y ) &\propto \int_{(0, \infty )^n} \bigg[ \pi( \beta ) \pi( \tau ) \frac{1}{\tau ^{n a}} \\
&\quad \times \prod_{i = 1}^{n} \left\{ {t_i}^{a + b - 1} e^{- t_i} {\nu _i}^{a - 1} e^{- t_i \nu _i / \tau } \frac{\beta^{\nu _i + 1} {\nu _i}^{\nu _i}}{\Gamma(\nu _i)} \frac{1} {{\lambda_i}^{\nu _i + 2}} e^{- \beta \nu _i / \lambda_i} \frac{1}{{\lambda_i}^{\delta_i}} e^{- ( \delta_i y_i ) /  \lambda_i } \right\} \bigg] \mathrm{d}t.
\end{align*}
We consider $t =(t_1, \dots, t_n) \in (0, \infty )^n$ as a set of additional latent variables. 
For the global parameters, we consider the conjugate gamma priors $\pi (\beta)=\mathrm{Ga}(\beta \mid a_{\beta},b_{\beta})$ and $\pi(\tau)=\mathrm{Ga}(\tau \mid a_{\tau},b_{\tau})$.

The variables $\lambda$, $\beta$, $\tau$, $t $, and $\nu $ are updated in the following way. 
\begin{itemize}
\item[-] 
Sample $\lambda_i \sim {\rm{IG}} ( \delta_i + \nu _i + 1, \delta_i y_i  + \beta \nu _i )$ independently for $i = 1, \dots , n$. 
\item[-] 
Sample $\beta \sim {\rm{Ga}} ( \sum_{i=1}^n \nu _i + n + a_{\beta } , \sum_{i=1}^n \nu _i / \lambda_i + b_{\beta } ) $. 
\item[-] 
Sample $\tau \sim {\rm{GIG}} ( - n a + a_{\tau } , 2 b_{\tau } , 2 \sum_{i=1}^n t_i \nu _i ) $, where ${\rm GIG}(a, b,\gamma)$ has density proportional to $x^{a - 1} \exp(- bx/2 - \gamma / 2x)$. 
\item[-] 
Sample $t_i \sim {\rm{Ga}} ( a + b, 1 + \nu _i/ \tau )$ independently for $i = 1, \dots , n$. 
\item[-]
The full conditional distribution of $\nu_i$ is proportional to 
\begin{align*}
&\prod_{i = 1}^{n} \{ {\rm{Ga}} ( \nu _i \mid a, t_i / \tau ) {\rm{Ga}} (1 / \lambda_i \mid \nu _i , \beta \nu _i ) \} , 
\end{align*}
which can be accurately approximated by using the method of \cite{Miller2019} for each $i = 1, \dots , n$. The method is based on the gamma approximation of intractable probability density function by matching the first- and second-derivatives of log densities. 
We use the approximate full conditional distributions as proposal distributions in independent Metropolis-Hastings (MH) steps. 
\end{itemize}

The full conditional distributions of parameters and latent variables other than $\nu_i$ are of familiar forms.
Even for the full conditional of $v_i$, we can efficiently sample from the distribution. 
Note that the number of latent variables in the proposed priors is larger than that of GL prior to exhibit global-local shrinkage properties. 
Hence, the computation time of the MCMC algorithm with the proposed priors can be longer than that of the GL prior.
Specifically, in the example given in Section~\ref{subsec:marginal_posterior}, the computation times of SB and IRB to generate 5000 posterior samples are around 5 seconds while that of GL is less than 1 second. 
Such an increase in computational costs would be a reasonable price for the desirable shrinkage properties.

\section{Theoretical properties}
\label{sec:theory}

In this section, we analytically compare properties of different priors for $u_i$ and, in particular, show two properties of the proposed priors, namely, tail-robustness for large observations (Section \ref{subsec:tail-robustness}) and desirable Kullback-Leibler risk bound under sparsity (Section \ref{subsec:KL}). 
For simplicity, we fix $\beta=\tau=1$ in what follows so that all the theoretical results are conditional on the hyperparameters.

\subsection{Tail-robustness %
for large observations }
\label{subsec:tail-robustness}

For a prior $\pi(u_i)$ of local parameter $u_i$, we consider the class given by 
\begin{align}
&\sup_{u \geqslant 1} \{ u \pi (u) \} < \infty, \label{eq:C2} \\
&\pi(u) \sim C \frac{u^{\alpha - 1}}{\{ 1 + \log (1 + 1 / u) \} ^{1 + \gamma }} \quad \text{as $u \to 0$ for some $\alpha \geqslant 0$ and $\gamma \geqslant - 1$}, \label{eq:C3} 
\end{align}
where $C$ is a positive constant. The notation $f(x)\sim g(x)$ means $\lim_{x\to 0} f(x)/ g(x)=1$. Condition \eqref{eq:C2} is a technical condition satisfied by most priors. 
Condition \eqref{eq:C3} is a condition on the tail of $\pi(u_i)$ at the origin and is satisfied by both the SB and the IRB priors. 
Because we consider proper distributions only in this paper, the case of $\alpha = 0$ and $\gamma \le 0$ is excluded. 

We consider the tail robustness of the 
Bayes estimator of $\lambda_i$ given by 
\begin{align*}
\hat{\lambda}_i &= \hat{\lambda}_{i}^{\rm{ML}} - E(\kappa_i \mid y_i) ( \hat{\lambda} _{i}^{\rm{ML}} - 1),
\end{align*}
where $\hat{\lambda}_{i}^{\rm{ML}} = y_i$ and $\kappa_i = u_i / ( \delta_i + u_i )$. 
Specifically, we show that the expected shrinkage factor, $E( \kappa_i \mid y_i )$, converges to zero as $y_i \to \infty $. 

\begin{thm}
\label{prp:tail_robustness} 
There exists a function $\kappa^{*}: (0, \infty ) \to (0, \infty )$ such that 
\begin{align*}
E(\kappa_i \mid y_i) &\sim \frac{1}{\delta_i} (1 + \alpha ) \kappa^{*} \left( \delta_i y_i \right) \to 0 
\end{align*}
as $y_i \to \infty $.
\end{thm}

Since the local parameter depends on not only the scale parameter but also the shape parameter, the evaluation of the posterior mean requires a detailed investigation of integrals involving gamma functions, where the details of the proof are given in the Supplementary Material. 
The constant $\alpha \geqslant 0$ is related to the tail of $\pi(u_i)$ at the origin. 
The heavier the tail is, the faster the expected shrinkage factor converges to zero. 
Finally, we note that if we fix $u_i = 1$, then $E(\kappa_i \mid y_i) = 1 / ( \delta_i + 1) $ does not converge to 0 as $y_i \to \infty $.

In the Supplementary Material, we further investigate the rate of $\kappa ^{*} (y)$, which shows that $\kappa ^{*} (y) = 1 / \log y$ as $y \to \infty $. 
This means that $E(\kappa_i \mid y_i)$ converges to $0$ very slowly as $y_i \to \infty $, while it remains a positive constant when $u_i=1$ even under $y_i\to \infty$.  
This property of $E(\kappa_i \mid y_i)$ indicates that $(E[\lambda_i \mid y_i]-y_i)/y_i\to 0 $ as $y_i\to\infty$, which is known as weakly tail-robust \citep{Hamura2019}.
Such property is also adopted to show the robustness of shrinkage under count response \citep{Datta2016} and correlated normal response \citep{okano2022locally}.

In the Supplementary Material, we also investigate the behavior of $E( \kappa _i \mid y_i )$ as $y_i \to 0$, where it is shown that $E( \kappa _i \mid y_i ) \to 0$ as $y_i \to 0$ if either $\pi ( u_i ) = \pi _{\rm{SB}} ( u_i )$ with $\delta_i \ge a$ or $\pi ( u_i ) = \pi _{\rm{IRB}} ( u_i )$.
This indicates that %
tail-robustness for a small observation is also established.  
\subsection{Kullback-Leibler super-efficiency under sparsity}
\label{subsec:KL}

We now consider the predictive efficiency for the proposed method \citep[e.g.][]{Polson2010, Carvalho2010, Datta2016}.  
In particular, we discuss the Kullback-Leibler divergence between the true sampling density and the Bayes predictive density under the proposed global-local shrinkage prior. We consider the following one-dimensional model
\begin{align*}
y\sim \mathrm{Ga} \left( \delta , \frac{\delta}{\lambda } \right) , \quad \lambda\sim \mathrm{IG}(1+u,u),\quad u\sim \pi(u).
\end{align*}
In the above model, let $f(y\mid\lambda) = \mathrm{Ga}(y\mid \delta , \delta / \lambda )$ and let $\lambda_0$ be the true value of $\lambda$. We define the Kullback-Leibler (KL) divergence between $f(y\mid\lambda)$ and $f(y\mid\lambda_0)$ by $D^{\mathrm{KL}}(\lambda_0,\lambda)=D^{\mathrm{KL}}(f(y\mid\lambda_0),f(y\mid\lambda))$. Then we have
\begin{align*}
D^{\mathrm{KL}}(\lambda_0,\lambda)= \delta \left( \frac{1 / \lambda}{1 / \lambda_0} -1-\log \frac{1 / \lambda}{1 / \lambda_0} \right) = \delta \left( \frac{\lambda_0}{\lambda } - 1 - \log \frac{\lambda_0}{\lambda} \right) .
\end{align*}
Furthermore, the KL neighborhood around $\lambda_0$ is defined by 
\begin{align*}
A_{\varepsilon}(\lambda_0)=\left\{\lambda \in (0,\infty)\mid D^{\mathrm{KL}}(\lambda_0,\lambda)<\varepsilon \right\}.
\end{align*}
We assume that the prior $p(\lambda)$ is information dense in the sense of $\mathrm{pr}(\lambda \in A_{\varepsilon}(\lambda_0))>0$ for all $\varepsilon>0$. 
From the Proposition 4 in \cite{Barron1987}, we have the Ces\'aro-mean risk $R_n$ is expressed by
\begin{align}\label{barron}
R_n\le \varepsilon - n^{-1} \log \mathrm{pr}(\lambda \in A_{\varepsilon}(\lambda_0)),
\end{align}
where $R_n=n^{-1}\sum_{k=1}^n D^{\mathrm{KL}}(f(y\mid\lambda_0)\mid\hat{f}_k(\lambda))$ and $\hat{f}_k(\lambda)$ is the Bayes predictive density under KL divergence using the posterior density based on $k\le n$ observations $y_1,\dots,y_k$. We now evaluate the prior probability $\mathrm{pr}(\lambda \in A_{\varepsilon}(\lambda_0))$ in the right-hand side of \eqref{barron} when $\lambda_0=1$.

Although we proved the theorem for the univariate case, the convergence in the multivariate case is derived from a component-wise application. 

\begin{thm}
\label{prp:KL} 
Assume that the true sampling model is $\mathrm{Ga}(\delta, \delta / \lambda _0 )$. 
For $\lambda_0\ne 1$, the Ces\'aro-mean risk for Bayes predictive density $\hat{f}_n$, which is the posterior mean of the density function $f(\cdot \mid\lambda)$, satisfies 
\begin{align*}
R_n = O\left( n^{- 1} \log n\right).
\end{align*}
If $\lambda_0 = 1$ and if $\pi(u) \propto u^{- 1 - b}$ as $u \to \infty $ for some $0 < b \le 1 / 2$, then 
\begin{align*}
R_n = O\left\{n^{-1}\left(\log n -\log \log n\right)\right\}. 
\end{align*}
\end{thm}

The proof of the theorem is given in the Supplementary Material. The results indicate that the Ces\'aro-mean risk achieves the optimal rate of convergence for the finite-dimensional parametric family when $\lambda_0 \ne 1$, while the risk has the super-efficient rate of Kullback-Leibler convergence for $\lambda_0=1$. The latter phenomenon is called {\it Kullback-Leibler super-efficiency}, which is a kind of higher-order optimality, and such results are commonly adopted to show theoretical superiority in handling sparsity in the context of global-local shrinkage priors \citep[e.g.][]{Polson2010, Carvalho2010, Datta2016}.
Theorem \ref{prp:KL} relates the right tail of $\pi(u_i)$ to the risk given in \eqref{barron}. 
To achieve Kullback-Leibler super-efficiency, it is sufficient to use $\pi(u_i)$ with a sufficiently heavy tail ($b \le 1/2$). 
Thus, $b$ plays a role in controlling sparsity at the grand mean. 
We remark that fixing $u_i = 1$ corresponds to using a point mass prior for $u_i$ and hence to violation of the sufficient condition that $\pi(u) \propto u^{-1-b}$ as $u \to \infty $.

\section{Simulation studies}
\label{sec:sim}

We evaluate the performance of Bayesian and frequentist shrinkage methods under gamma response. 
Let $y_i\sim \mathrm{Ga}(\delta_i, \delta_i/\lambda_i)$ for $i=1,\ldots,n(=200)$ and $\delta_i=5$. 
We consider the following six scenarios of the true mean $\lambda_i$:
\begin{align*}
\text{(Scenario 1)}& \ \ \lambda_i\sim 0.95 \delta_\mu + 0.05 {\rm Ga}(20\mu, 2), \ \ \ \ 
\text{(Scenario 2)} \ \ \lambda_i\sim 0.9 \delta_\mu + 0.1 {\rm Ga}(20\mu, 2), \\
\text{(Scenario 3)}& \ \ \lambda_i\sim 0.95 \delta_\mu + 0.05 \mu |t_3|, \ \ \ \ 
\text{(Scenario 4)} \ \ \lambda_i\sim 0.9 {\rm Ga}(5\mu, 5) + 0.1 \mu|t_1|,\\
\text{(Scenario 5)}& \ \ \lambda_i\sim 0.9 \delta_\mu + 0.1 {\rm Ga}(10\mu, 2), \ \ \ \ 
\text{(Scenario 6)} \ \ \lambda_i\sim 0.85 \delta_\mu + 0.15 {\rm Ga}(10\mu, 2), 
\end{align*}
where $\mu=5$, $\delta_a$ denotes a point mass at $a$ and $t_c$ denotes a $t$-distribution with $c$ degrees of freedom. In the first three scenarios, most of the true means $\lambda_i$ is exactly equal to $\mu=5$, and a small part of true means are very large compared with $\mu$. 
In scenario 4, most true means are concentrated around $\mu$ (not exactly equal to 0).

For the simulated data, we apply six methods, the proposed scaled beta (SB) and inverse rescaled beta (IRB) priors, the global shrinkage (GL) prior (setting $u_i=1$ in the proposed model), shrinkage estimators given by \cite{DasGupta1986}, denoted by DG, adaptive variance shrinkage estimators by \cite{Lu2016}, denoted by VS, and maximum likelihood (ML) estimator $y_i$.
Note that the DG method is to provide decision-theoretic point estimates of $\lambda_i$ by minimizing a weighted quadratic loss function, and the VS method uses a finite mixture of inverse-gamma distributions as a prior distribution for $\lambda_i$. 
The tuning parameters in the SB and IRB priors are set to $a=2$ and $b=1/2$. 
We used non-informative gamma priors, $\beta\sim {\rm Ga(0.1, 0.1)}$ and $\tau\sim {\rm Ga(0.1, 0.1)}$ for SB, IRB and GL priors. 
For the Bayesian methods, 3000 posterior samples are generated after discarding the first 2000 samples as burn-in. 
Note that we used the R package ``vashr" (\url{https://github.com/mengyin/vashr}) to apply the VS method, where the degrees of freedom of $\chi^2$-distribution is set to $2\delta_i(=10)$.

We first investigate the shrinkage property of the proposed global-local shrinkage priors compared with the other methods. 
In Figure~\ref{fig:sim-shrink}, we show scatter plots of observed values and point estimates (posterior means for the Bayesian methods) produced by five shrinkage methods under scenario 1. 
It is observed that the standard shrinkage methods, GL, DG, and VS, linearly shrink the observed value $y_i$, that is, the shrinkage factor is constant regardless of $y_i$.
On the other hand, the proposed SB and IRB priors more strongly shrink the observed values around $\lambda_i=5$, showing the adaptive shrinkage property of the global-local shrinkage prior.

\begin{figure}[!htb]
\centering
\includegraphics[width = \linewidth]{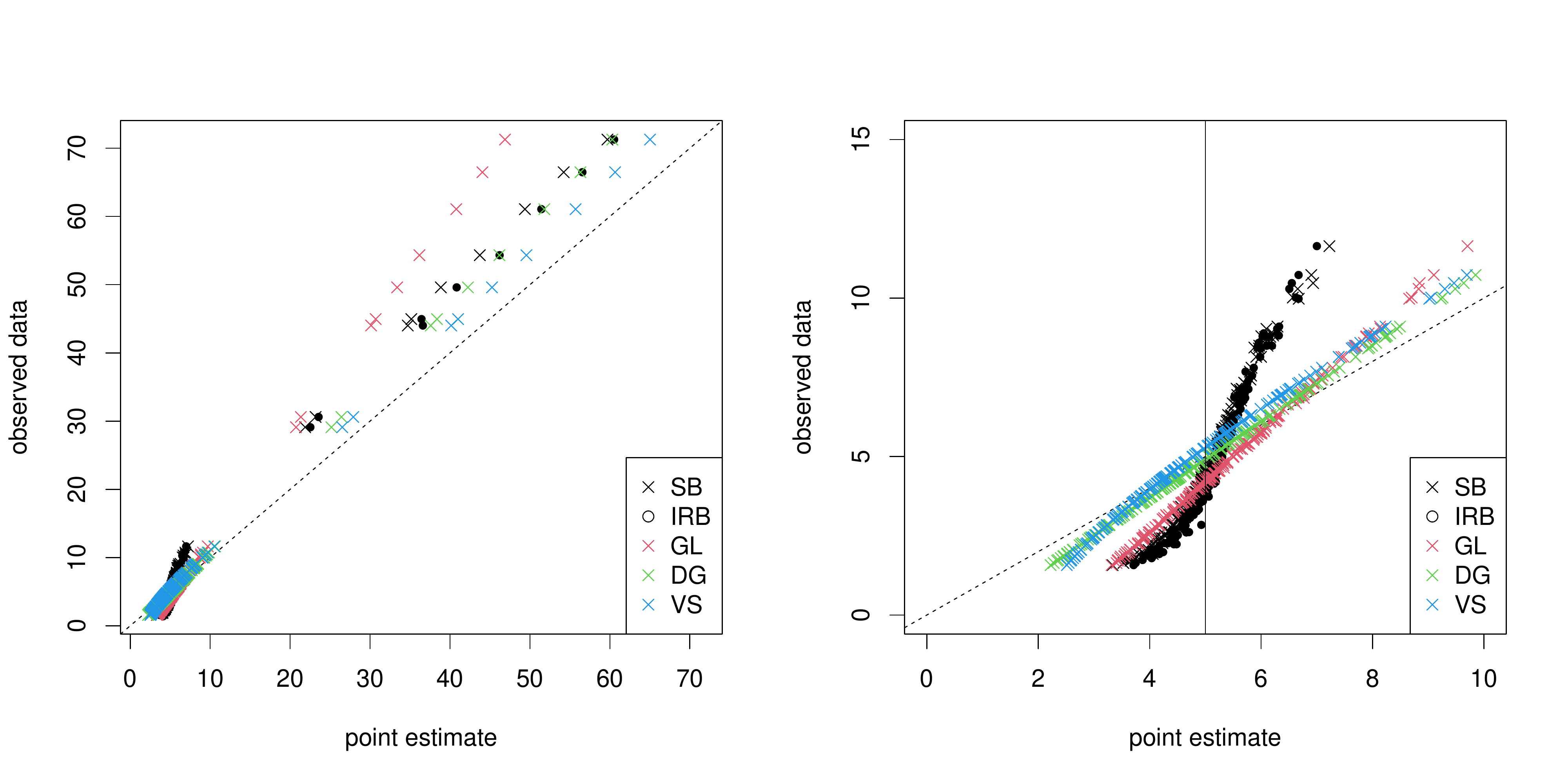}
\caption{Scatter plots of observed values (ML) and point estimates obtained from five shrinkage methods. 
The right panel is an enlarged version of the left panel.
The vertical line in the right panel indicates the location of null signals.  }
\label{fig:sim-shrink}
\end{figure}%

We next evaluate mean absolute percentage error (MAPE), defined as $n^{-1}\sum_{i=1}^n\lambda_i^{-1} | \lambda_i-\widehat{\lambda}_i |$ with a point estimate $\widehat{\lambda}_i$.
We present boxplots of MAPE for 1000 replications in Figure~\ref{fig:sim}.
The results indicate that the proposed SB and IRB provide more accurate point estimates than the other methods in all the scenarios, except for IRB under Scenario 3. 
The amount of improvement of the proposed methods is remarkable when the null and non-null signals are well-separated, as in Scenarios 1 and 2.
Comparing SB and IRB, SB tends to provide a smaller overall MAPE than IRB.
To compare the two methods more precisely, we also computed MAPE only for non-null signals.
The averaged values of MAPE for non-null signals are given in Table~\ref{tab:sim-mape}, which shows that IRB performs slightly better than SB for the estimation of non-null signals, and this is consistent with the stronger tail-robustness property of IRB than that of SB.

Furthermore, we computed the coverage probability (CP) and average length (AL) of $95\%$ credible/confidence intervals.
We only consider the ML method for the frequentist methods since DG and VS do not provide interval estimation. 
The $95\%$ confidence interval of ML can be obtained as $(y_i/{\rm P}_{G}(0.975; \delta_i,\delta_i), y_i/{\rm P}_{G}(0.025; \delta_i,\delta_i))$, where ${\rm P}_G(\cdot;\alpha,\beta)$ denotes the probability function of ${\rm Ga}(\alpha, \beta)$.
The CP and AL averaged over 1000 Monte Carlo replications are given in Table~\ref{tab:sim}.
It can be seen that all three Bayesian methods have empirical CP values larger than the nominal level of $0.95$ except in Scenario 6, whereas the interval lengths of SB and IRB tend to be smaller than GL and ML in all the scenarios.

\begin{figure}[!htb]
\centering
\includegraphics[width = \linewidth]{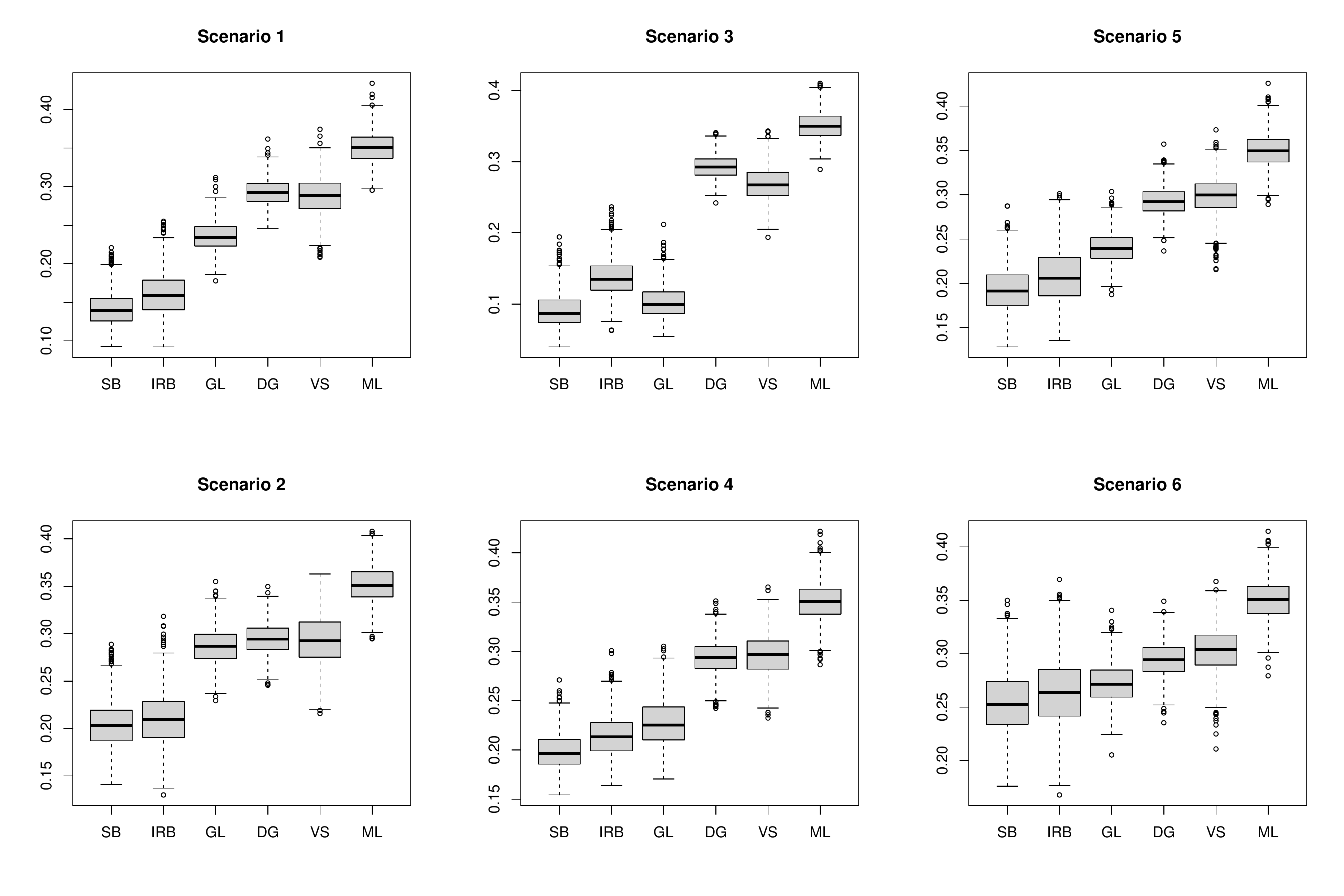}
\caption{Boxplots of mean absolute percentage errors (MAPE) for 1000 Monte Carlo replications. }
\label{fig:sim}
\end{figure}%

\begin{table}[!htb]
\caption{Mean absolute percentage errors (MAPE) for non-null signals averaged over 1000 Monte Carlo replications. }
\label{tab:sim-mape}
\begin{center}
\begin{tabular}{cccccccccccc}
\hline
Scenario &  & 1 & 2 & 3 & 4 & 5 & 6 \\
\hline
SB &  & 0.365 & 0.359 & 0.304 & 0.292 & 0.388 & 0.384 \\
IRB &  & 0.360 & 0.360 & 0.284 & 0.286 & 0.385 & 0.386 \\
\hline
\end{tabular}
\end{center}
\end{table}

\begin{table}[!htb]
\caption{Coverage probabilities and average lengths of $95\%$ credible/confidence intervals averaged over 1000 Monte Carlo replications. }
\label{tab:sim}
\begin{center}
\begin{tabular}{cccccccccccc}
\hline
 &  &  \multicolumn{4}{c}{Coverage probability} && \multicolumn{4}{c}{Average length } \\
Scenario &  & SB & IRB & GL & ML &  & SB & IRB & GL & ML \\
\hline
1 &  & 98.7 & 98.4 & 97.6 & 95.0 &  & 1.13 & 1.19 & 1.58 & 2.59 \\
2 &  & 97.3 & 97.1 & 97.1 & 95.0 &  & 1.32 & 1.33 & 1.79 & 2.59 \\
3 &  & 98.1 & 98.4 & 97.9 & 95.0 &  & 0.81 & 1.09 & 0.91 & 2.59 \\
4 &  & 95.6 & 96.0 & 97.2 & 95.0 &  & 1.12 & 1.23 & 1.43 & 2.59 \\
5 &  & 96.6 & 96.3 & 96.7 & 95.0 &  & 1.23 & 1.27 & 1.56 & 2.59 \\
6 &  & 94.4 & 94.0 & 96.2 & 95.0 &  & 1.37 & 1.39 & 1.69 & 2.59 \\

\hline
\end{tabular}
\end{center}
\end{table}

\section{Real data example}
\label{sec:data}

\subsection{Average admission period of COVID-19 in Korea}
\label{sec:covid}

We first apply the global-local shrinkage techniques to estimate the average length of hospital stay of COVID-19-infected persons. We use the data set available at Kaggle (\url{https://www.kaggle.com/kimjihoo/coronavirusdataset}), where the date of admission and discharge is observed for 1587 individuals in Korea. 
We then group these individuals regarding 98 cities and three classes of age, young (39 or less), middle (from 40 to 69), and old (70 or more), resulting in $n=185$ groups after omitting empty groups. 
Assuming exponential distributions with group-specific mean for admission period (days) of each individual, the group-wise sample mean is distributed as ${\rm Ga}(n_i, n_i/\lambda_i)$ for $i=1,\ldots,n$, where $n_i$ is the number of individuals within the $i$th group and $\lambda_i$ is the true mean of admission period specific to the $i$th group. 
Note that $n_i$ ranges from 1 to 258, and the scatter plot of $n_i$ and $y_i$ are given in Figure \ref{fig:covid-plot}.

We apply the proposed SB prior as well as GL and DG methods. 
Regarding the prior distributions for grand mean $\beta$ in the SB and GL models, we assign a non-informative prior, ${\rm Ga}(0.1, 0.1)$.
Furthermore, we set $a=2$ and $b=1/2$ in the proposed prior distributions. 
The posterior means for SB and GL are computed based on 10000 posterior samples (after discarding 3000 samples), whose histograms are shown in Figure \ref{fig:covid}. 
The posterior mean of the grand mean $\beta$ in the SB model was $22.0$ ($95\%$ credible interval was $(21.3, 22.8)$), which is consistent with the evidence that the average admission period is around $21$ \citep[e.g.][]{jang2021comorbidities}.
On the other hand, the posterior mean of the grand mean $\beta$ in the GL model was $24.4$, where $95\%$ credible interval was $(22.5, 26.6)$.
We also present the histogram of shrinkage estimates made by DG. 
It is observed that SB strongly shrinks the observed values toward the grand mean $\beta$ so that most of the posterior means of the average admission period are concentrated around the grand mean.
This is because most groups having large sample means have small sample sizes, and such unreliable information is strongly shrunk.
We found that only a single group (old age class of Gyeongsan-si) has a much larger average admission period, about 35 days. 
Since the sample size of this group is $107$, and the sample mean is about 37, the posterior result seems reasonable. 
To see more detailed results, we present scatter plots of observed values and posterior means against sample size $n_i$, in Figure \ref{fig:covid-plot}.
It is observed that the amount of shrinkage (i.e., the difference between observed values and posterior means) decreases as $n_i$ increases, and observations having small sample sizes strongly shrunk toward the grand mean. 
From Figure \ref{fig:covid}, it can also be seen that GL also provides reasonably shrunk estimates of $\lambda_i$ and DG does not, but the proposed SB prior can provide strongly shrunk point estimates. 
Moreover, the average length of $95\%$ credible intervals made by SB was 19.0, which was considerably smaller than the 22.9 produced by GL.

\begin{figure}[!htb]
\centering
\includegraphics[width = \linewidth]{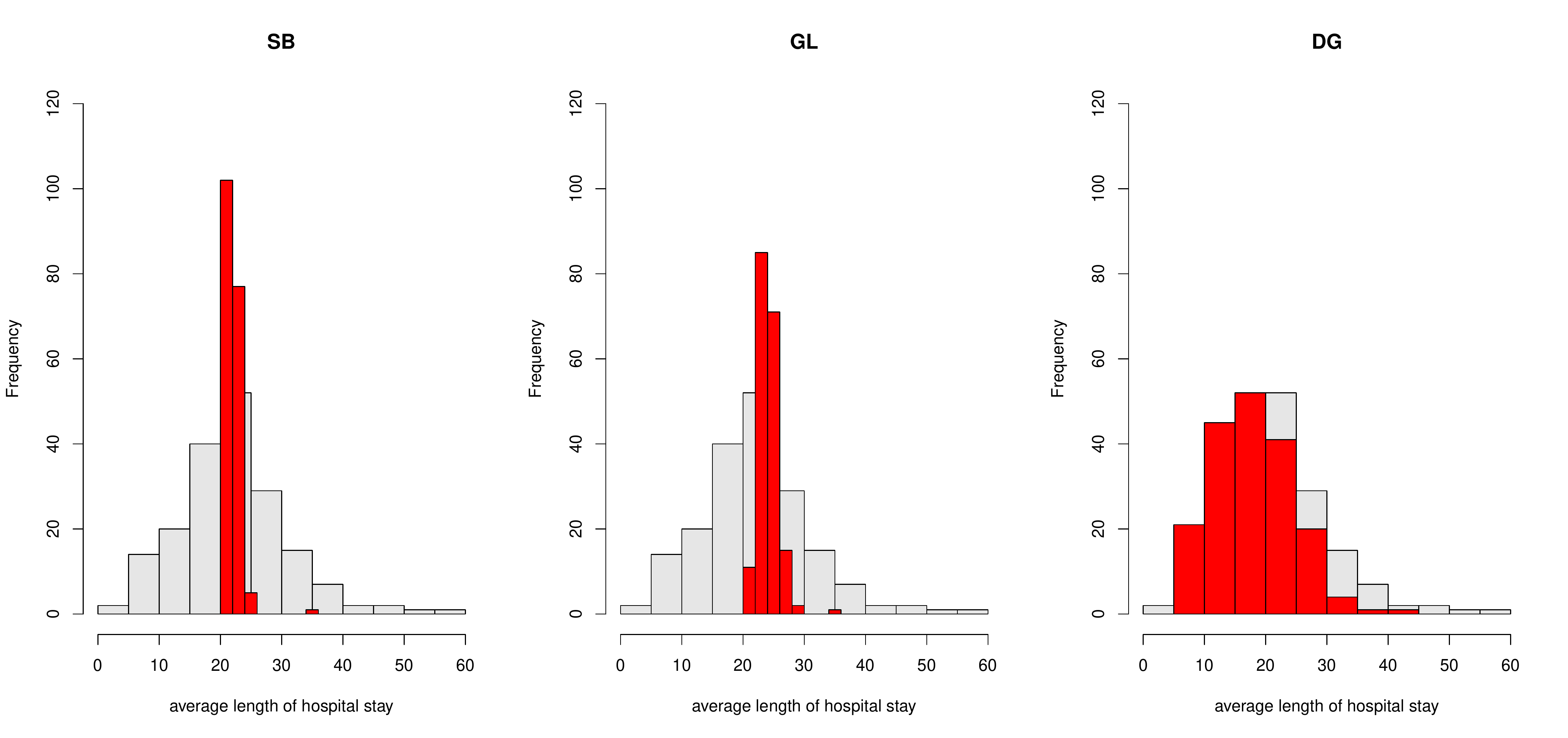}
\caption{Histograms of shrinkage estimates (red) and observed values (grey) of the average admission period. }
\label{fig:covid}
\end{figure}%

\begin{figure}[!htb]
\centering
\includegraphics[width = \linewidth]{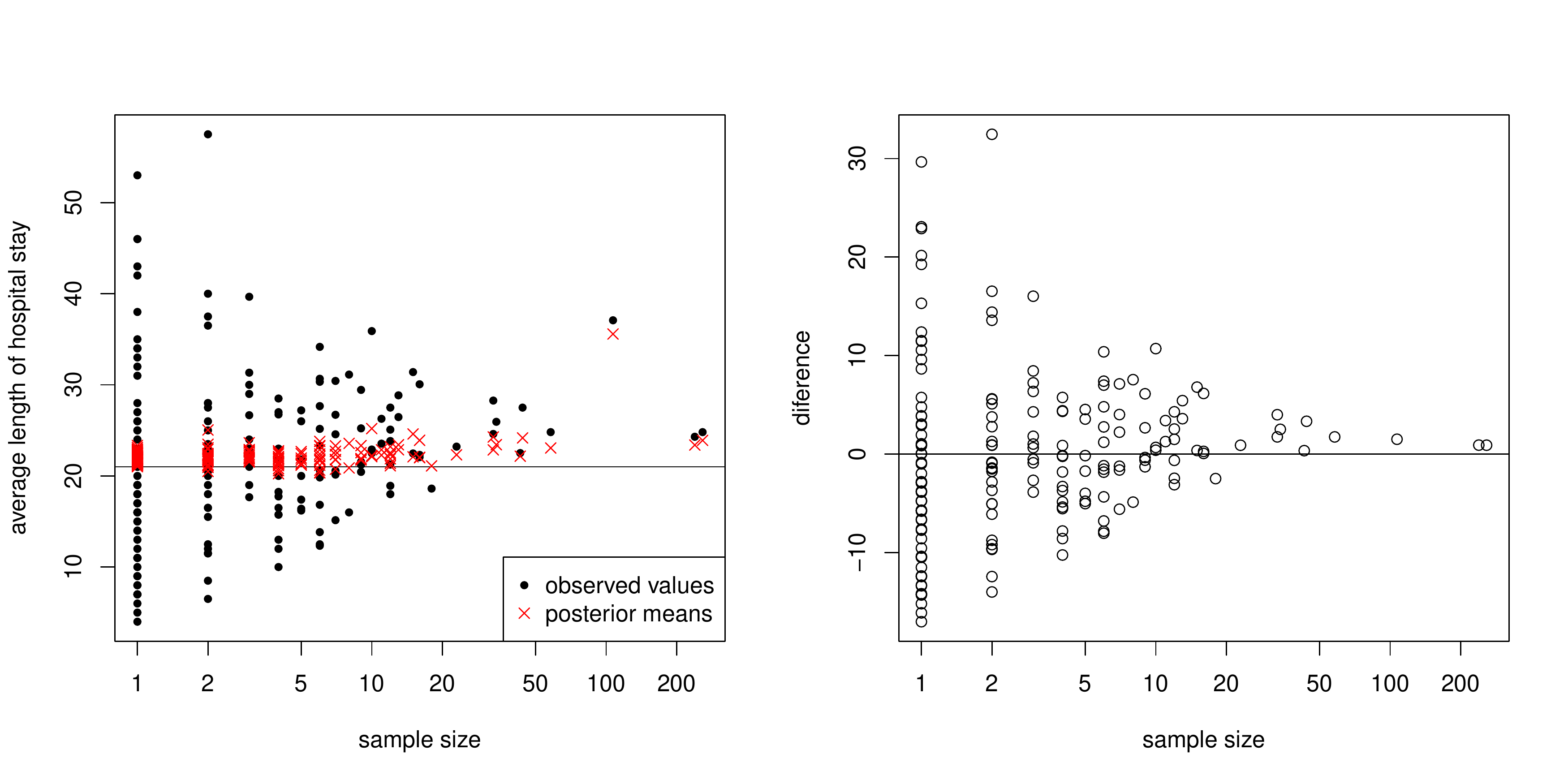}
\caption{Left: Scatter plot of sample size $n_i$ and average admission.
Right: Scatter plot of sample size $n_i$ and difference of $y_i$ and posterior mean of $\lambda_i$.}
\label{fig:covid-plot}
\end{figure}%

\subsection{Variance estimation of gene expression data}
\label{sec:gene}

We next apply the shrinkage methods to variance estimation of gene expression data. 
As noted in \cite{Lu2016}, in gene expression analysis aimed at identifying deferentially expressed genes, accurate estimation of the unknown variance is an essential step since it directly relates to the degree of statistical significance. 
We use a popular prostate cancer dataset from
\cite{singh2002gene}. 
In this dataset, there are gene expression values for $n=6033$ genes for $50$ subjects in control subjects.
We compute sampling variances of $n$ gene expressions, distributed as ${\rm Ga}( n_i/2, n_i/2\lambda_i)$ for $i=1,\ldots,n$, where $\lambda_i$ is the true variance of the $i$th gene expression. 
By assigning non-informative priors, ${\rm Ga}(0.1, 0.1)$ for $\beta$ in the SB, IRB, and GL models as well as the VS method. 
In the Bayesian methods, we computed posterior means using 2000 posterior samples after discarding the first 1000 samples.
The histograms of posterior means are shown in Figure~\ref{fig:gene}.
As confirmed in the previous example, we can see that the proposed SB and IRB priors can provide more shrunk estimates than the other methods. 
Furthermore, average lengths of $95\%$ credible intervals made by SB and IRB were 0.640 and 0.639, respectively, which are smaller than 0.663 by GL.
This shows the efficiency of the proposed priors.

\begin{figure}[!htb]
\centering
\includegraphics[width = \linewidth]{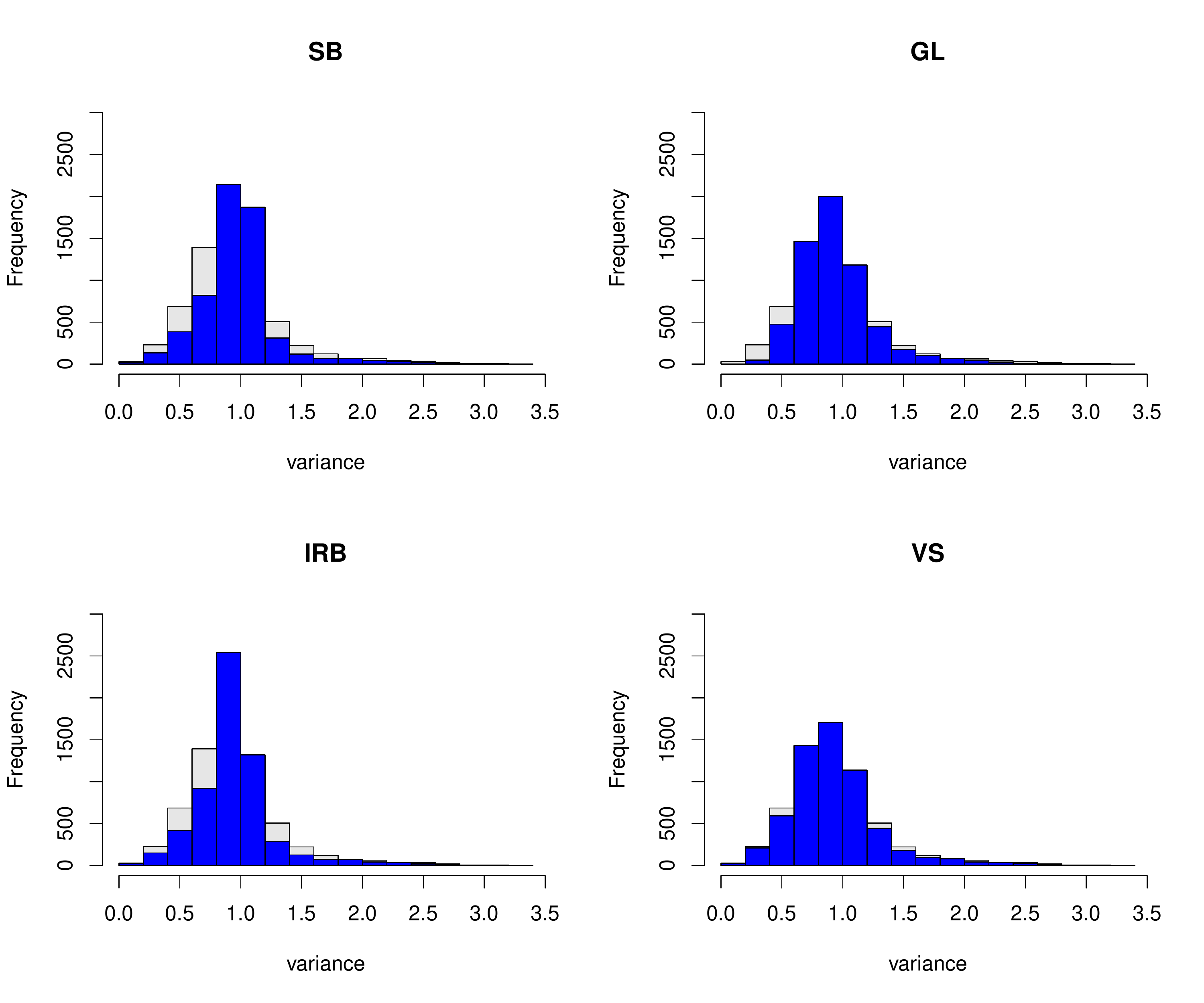}
\caption{Histograms of shrinkage estimates (blue) and observed values (grey) of variances of gene expression data. }
\label{fig:gene}
\end{figure}%

\section{Discussion}
\label{sec:discussion}

We proposed a new class of continuous global-local shrinkage priors for high-dimensional positive-valued parameters based on shape-scale mixtures of inverse-gamma distributions.
Although this paper focuses on a sequence of gamma-distributed observations, it can also be useful in other models.
One notable application would be using a flexible error distribution in a regression model for gamma-distributed observations, such as gamma regression or accelerated failure time models.
For the latter model, the proposed distribution may cast an alternative to the Bayesian nonparametric approach \citep[e.g.][]{hanson2006modeling,kuo1997bayesian}, and some comparisons would be an interesting future study. 
Although the approximate sampling method of \cite{Miller2019} is adopted in sampling from the local parameter in the proposed priors, it might be worth implementing the more recent data augmentation technique of \cite{hamura2022data}.

\section*{Acknowledgments}
This work is partially supported by Japan Society for Promotion of Science (KAKENHI) grant numbers 22K20132, 20J10427, 19K11852, 21K13835, and 21H00699.

\vspace{5mm}
\bibliographystyle{chicago}
\bibliography{ref}

\newpage

\newpage
\setcounter{page}{1}
\setcounter{equation}{0}
\renewcommand{\theequation}{S\arabic{equation}}
\setcounter{section}{0}
\renewcommand{\thelem}{S\arabic{lem}}
\setcounter{thm}{0}
\renewcommand{\thethm}{S\arabic{thm}}
\setcounter{prp}{0}
\renewcommand{\theprp}{S\arabic{prp}}
\setcounter{table}{0}
\renewcommand{\thesection}{S\arabic{section}}
\setcounter{table}{0}
\renewcommand{\thetable}{S\arabic{table}}
\setcounter{figure}{0}
\renewcommand{\thefigure}{S\arabic{figure}}

\begin{center}
{\LARGE\bf Supplementary Materials for ``Sparse Bayesian inference on gamma-distributed observations using shape-scale inverse-gamma mixtures"}
\end{center}

\begin{center}
{\large Yasuyuki Hamura$^{1}$, Takahiro Onizuka$^{2}$,\\ Shintaro Hashimoto$^{2}$ and Shonosuke Sugasawa$^{3}$}
\end{center}

\medskip
\noindent
$^{1}$Graduate School of Economics, Kyoto University\\
$^{2}$Department of Mathematics, Hiroshima University\\
$^{3}$Center for Spatial Information Science, The University of Tokyo

\vspace{1cm}
This Supplementary Material provides 
proofs and technical details related to the main text.

\section{Preliminaries}
The following facts %
will be used later in this Supplementary Material. 
\begin{itemize}
\item
As $x \to 0$, $\Ga (x) \sim 1 / x$. 
As $x \to \infty $, $\Ga (x) \sim (2 \pi )^{1 / 2} x^{x - 1 / 2} e^{- x}$. 
\item
For any $x > 0$, 
\begin{align}
&{e^{- 1 / (12 x)} \over (2 \pi )^{1 / 2}} < {x^{x - 1 / 2} \over \Ga (x) e^x} < {1 \over (2 \pi )^{1 / 2}} \text{,} \non \\
&\log x - {{\Ga }' (x) \over \Ga (x)} - {1 / 2 \over x} = \log x - \psi (x) - {1 / 2 \over x} = 2 \int_{0}^{\infty } {1 \over t^2 + x^2} {t \over e^{2 \pi t} - 1} \mathrm{d}t \text{.} \non %
\end{align}
\item
Let $c \in \mathbb{R}$. 
Let $\Psi \colon (0, \infty ) \to [0, \infty )$ be a continuously differentiable function. 
Suppose that $x {\Psi }' (x) / \Psi (x) \to c$ as $x \to 0$ (as $x \to \infty $). 
Then, for any $v > 0$, 
\begin{align}
{\Psi (v x) \over \Psi (x)} \to v^c \non 
\end{align}
as $x \to 0$ (as $x \to \infty $). 
\end{itemize}

\section{Tail properties of the marginal prior for $\la _i$}
Here, we investigate properties of the tails of the marginal prior $p( \la _i )$ of $\la _i$ under the proper prior $u_i \sim \pi ( u_i )$. 

\begin{thm}%
\label{prp:marginal_general} 
Let $\xi _i = 1 / \la _i - 1 - \log (1 / \la _i ) > 0$. 
The marginal prior of $\la _i$ has the following properties. 
\begin{itemize}
\setlength{\leftskip}{12pt}
\item[{\rm{(i)}}] 
Suppose that there exist $\alpha \geqslant 0$ and $\gamma \geqslant - 1$ such that $\pi (u) \propto u^{\alpha - 1} / \{ 1 + \log (1 + 1 / u) \} ^{1 + \gamma }$ as $u \to 0$. 
Then, as $\la _i \to 0$ and $\la _i \to \infty $, we have $\xi _i \to \infty $ and 
\begin{align*}
p( \la _i ) &\sim \Ga ( \al + 1) {1 \over {\la _i}^2} {1 \over {\xi _i}^2} \pi \left( {1 \over \xi _i} \right) \propto {1 \over {\la _i}^2} {1 \over {\xi _i}^{\al + 1}} {1 \over \{ 1 + \log (1 + \xi _i ) \} ^{1 + \ga }}.
\end{align*}
\item[{\rm{(ii)}}]
Suppose that there exists $b > 0$ such that $\pi (u) \propto u^{- 1 - b}$ as $u \to \infty $. 
Then, as $\la _i \to 1$, we have $\xi _i \to 0$ and 
\begin{align*}
p( \la _i ) &\begin{cases} \displaystyle \sim {\Ga (1 / 2 - b) \over (2 \pi )^{1 / 2}} \left( {1 \over \xi _i} \right) ^{1 + 1 / 2} \pi \left( {1 \over \xi _i} \right) \propto \left( {1 \over \xi _i} \right) ^{1 / 2 - b} \to \infty \text{,} & \text{if $b < 1 / 2$}, \\ \displaystyle \to \infty \text{,} & \text{if $b = 1 / 2$} \text{,} \\ \displaystyle \sim \int_{0}^{\infty } \pi (u) {u^u e^{- u} \over \Ga (u)} du < \infty \text{,} & \text{if $b > 1 / 2$}.\end{cases}
\end{align*}
\item[{\rm{(iii)}}]
Suppose that there exists $a > 0$ such that $\pi (u) = \{ 1 / B(a, 1 / 2) \} u^{a - 1} / (1 + u)^{a + 1 / 2}$ for all $u \in (0, \infty )$. 
Then, as $\la _i \to 1$, we have $\xi _i \to 0$ and 
\begin{align*}
p( \la _i ) \sim {1 \over (2 \pi )^{1 / 2}} \left( \log {1 \over \xi _i} \right) \left( {1 \over \xi _i} \right) ^{3 / 2} \pi \left( {1 \over \xi _i} \right) \propto \log {1 \over \xi _i} \to \infty.
\end{align*}
\item[{\rm{(iv)}}]
Suppose that there exists $a > 0$ such that $\pi (u) = \{ 1 / B(1 / 2, a) \} [1 / \{ u (1 + u) \} ] \{ \log (1 + 1 / u) \} ^{1 / 2 - 1} / \{ 1 + \log (1 + 1 / u) \} ^{1 / 2 + a}$ for all $u \in (0, \infty )$. 
Then, as $\la _i \to 1$, we have $\xi _i \to 0$ and 
\begin{align*}
p( \la _i ) \sim {1 \over (2 \pi )^{1 / 2}} \left( \log {1 \over \xi _i} \right) \left( {1 \over \xi _i} \right) ^{3 / 2} \pi \left( {1 \over \xi _i} \right) \propto \log {1 \over \xi _i} \to \infty , 
\end{align*}
which is exactly as in part (iii). 
\end{itemize}
\end{thm}

For part (i), the SB and IRB priors correspond to setting $\al = a$ and $\ga = - 1$ and setting $\al = 0$ and $\ga > 0$, respectively. 
Part (ii) is directly applicable to both the SB and IRB priors. 
Although part (ii) does not consider the speed with which $p( \la _i )$ tends to infinity as $\la _i \to 1$ when $b = 1 / 2$, this boundary case is treated in parts (iii) and (iv) for the SB and IRB priors with $b = 1 / 2$, respectively. 

In deriving Proposition 1 in the main manuscript from Theorem \ref{prp:marginal_general}, we note that 
\begin{align}
\xi _i \sim \begin{cases} 1 / \la _i \text{,} & \text{as $\la _i \to 0$} \text{,} \\ \log \la _i \text{,} & \text{as $\la _i \to \infty $} \text{,} \\ ( \la _i - 1)^2 / 2 \text{,} & \text{as $\la _i \to 1$} \text{.} \end{cases} \non 
\end{align}
It follows from parts (iii) and (iv) that 
\begin{align}
p( \la _i ) &\propto \log {1 \over | \la _i - 1|} \to \infty \non 
\end{align}
as $\la _i \to 1$ for the SB and IRB priors with $b = 1 / 2$.

\begin{proof}[Proof of Theorem \ref{prp:marginal_general}]
Let $\xi _i = 1 / \la _i - 1 - \log (1 / \la _i ) \geqslant 0$ and let 
\begin{align}
f(u) = {u^{u - 1} e^{- u} \over \Ga (u)} \non 
\end{align}
for $u \in (0, \infty )$. 
Then the marginal density can be written as 
\begin{align}
p( \la _i ) &= {1 \over {\la _i}^2} \int_{0}^{\infty } \pi (u) {u^u \over \Ga (u)} {1 \over {\la _i}^u} e^{- u / \la _i} \mathrm{d}u = {1 \over {\la _i}^2} \int_{0}^{\infty } u \pi (u) f(u) e^{- u \xi _i} \mathrm{d}u \text{.} \non 
\end{align}

For part (i), 
\begin{align}
\left| {\int_{0}^{\infty } u \pi (u) f(u) e^{- u \xi _i} \mathrm{d}u \over \int_{0}^{1} u \pi (u) f(u) e^{- u \xi _i} \mathrm{d}u } - 1 \right| &\le {\int_{1}^{\infty } u \pi (u) f(u) e^{- u \xi _i} \mathrm{d}u \over \int_{0}^{1 / 2} u \pi (u) f(u) e^{- u \xi _i} \mathrm{d}u } \non \\
&\le {\int_{1}^{\infty } u \pi (u) f(u) e^{- u} \mathrm{d}u \over \int_{0}^{1 / 2} u \pi (u) f(u) e^{- u} \mathrm{d}u } {e^{- ( \xi _i - 1)} \over e^{- (1 / 2) ( \xi _i - 1)}} \to 0 \non 
\end{align}
as $1 < \xi _i \to \infty $. 
Therefore, as $1 / \la _i$ or $\la _i$ tends to infinity, we have $\xi _i \to \infty $ and 
\begin{align}
{{\la _i}^2 p( \la _i ) \over \pi (1 / \xi _i ) / {\xi _i}^2} &\sim {1 \over \pi (1 / \xi _i ) / {\xi _i}^2} \int_{0}^{1} u \pi (u) f(u) e^{- u \xi _i} \mathrm{d}u \non \\
&= {1 \over \pi (1 / \xi _i ) / {\xi _i}^2} \int_{0}^{\xi _i} {1 \over \xi _i} {u \over \xi _i} \pi (u / \xi _i ) f(u / \xi _i ) e^{- u} \mathrm{d}u \non \\
&= \int_{0}^{\infty } u e^{- u} {1(u < \xi _i ) \pi (u / \xi _i ) \over \pi (1 / \xi _i )} 1(u < \xi _i ) f(u / \xi _i ) \mathrm{d}u \text{.} \non 
\end{align}
Note that 
\begin{align}
&\infty > {\sup_{u \in (0, 1)} ( \pi (u) / [ u^{\al - 1} / \{ 1 + \log (1 + 1 / u) \} ^{1 + \ga } ]) \over \inf_{u \in (0, 1)} ( \pi (u) / [ u^{\al - 1} / \{ 1 + \log (1 + 1 / u) \} ^{1 + \ga } ])} \geqslant {1(u < \xi _i ) \pi (u / \xi _i ) \over \pi (1 / \xi _i )} \to u^{\al - 1} \non 
\end{align}
for all $u \in (0, \infty )$ since $1 + \log (1 + \xi _i )$ is a slowly varying function of $\xi _i \to \infty $ and that 
\begin{align}
&\infty > \sup_{u \in (0, 1)} {u^u \over \Ga (1 + u)} \geqslant 1(u < \xi _i ) {(u / \xi _i )^{u / \xi _i} e^{- u / \xi _i} \over \Ga (1 + u / \xi _i )} = 1(u < \xi _i ) f(u / \xi _i ) \to 1 \non 
\end{align}
for all $u \in (0, \infty )$. 
Then, by the dominated convergence theorem, 
\begin{align}
{{\la _i}^2 p( \la _i ) \over \pi (1 / \xi _i ) / {\xi _i}^2} &\sim \int_{0}^{\infty } u^{\al } e^{- u} \mathrm{d}u = \Ga ( \al + 1) \text{,} \non 
\end{align}
and this proves part (i). 

For part (ii), suppose first that $b < 1 / 2$. 
Then 
\begin{align}
p( \la _i ) &\sim \int_{0}^{\infty } u \pi (u) f(u) e^{- u \xi _i} \mathrm{d}u = \int_{0}^{1} u \pi (u) f(u) e^{- u \xi _i} \mathrm{d}u + \int_{1}^{\infty } u \pi (u) f(u) e^{- u \xi _i} \mathrm{d}u \non \\
&= \int_{0}^{1} u \pi (u) f(u) e^{- u \xi _i} \mathrm{d}u + \int_{\xi _i}^{\infty } {1 \over \xi _i} {u \over \xi _i} \pi \Big( {u \over \xi _i} \Big) f \Big( {u \over \xi _i} \Big) e^{- u} \mathrm{d}u \text{.} \non 
\end{align}
By the monotone convergence theorem, 
\begin{align}
\int_{0}^{1} u \pi (u) f(u) e^{- u \xi _i} \mathrm{d}u &\to \int_{0}^{1} u \pi (u) f(u) \mathrm{d}u \le \Big[ \sup_{u \in (0, 1)} \{ u f(u) \} \Big] \int_{0}^{1} \pi (u) \mathrm{d}u < \infty \text{.} \non 
\end{align}
Meanwhile, 
\begin{align}
&\int_{\xi _i}^{\infty } {1 \over \xi _i} {u \over \xi _i} \pi \Big( {u \over \xi _i} \Big) f \Big( {u \over \xi _i} \Big) e^{- u} \mathrm{d}u / \Big\{ {1 \over {\xi _i}^2} \pi \Big( {1 \over \xi _i} \Big) f \Big( {1 \over \xi _i} \Big) \Big\} \non \\
&= \int_{\xi _i}^{\infty } {1 \over u^{1 / 2 + b}} {(u / \xi _i )^{1 + b} \pi (u / \xi _i ) \over (1 / \xi _i )^{1 + b} \pi (1 / \xi _i )} {(u / \xi _i )^{1 / 2} f(u / \xi _i ) \over (1 / \xi _i )^{1 / 2} f(1 / \xi _i )} e^{- u} \mathrm{d}u \text{.} \non 
\end{align}
For all $u \in (0, \infty )$, we have that 
\begin{align}
\infty > \frac{ \sup_{u \in (1, \infty )} \{ u^{1 + b} \pi (u) \} }{ \inf_{u \in (1, \infty )} \{ u^{1 + b} \pi (u) \} } \geqslant {(u / \xi _i )^{1 + b} \pi (u / \xi _i ) \over (1 / \xi _i )^{1 + b} \pi (1 / \xi _i )} \to 1 \non 
\end{align}
by assumption and also that 
\begin{align}
{(u / \xi _i )^{1 / 2} f(u / \xi _i ) \over (1 / \xi _i )^{1 / 2} f(1 / \xi _i )} &= \frac{ \displaystyle {(u / \xi _i )^{u / \xi _i - 1 / 2} e^{- u / \xi _i} \over \Ga (u / \xi _i )} }{ \displaystyle {(1 / \xi _i )^{1 / \xi _i - 1 / 2} e^{- 1 / \xi _i} \over \Ga (1 / \xi _i )} } \begin{cases} \displaystyle \le \exp \Big( {1 \over 12 / \xi _i} \Big) \downarrow 1 \\ \displaystyle \geqslant \exp \Big( - {1 \over 12 u / \xi _i} \Big) \uparrow 1 \text{.} \end{cases} \non 
\end{align}
Therefore, by the dominated convergence theorem, 
\begin{align}
\int_{\xi _i}^{\infty } {1 \over u^{1 / 2 + b}} {(u / \xi _i )^{1 + b} \pi (u / \xi _i ) \over (1 / \xi _i )^{1 + b} \pi (1 / \xi _i )} {(u / \xi _i )^{1 / 2} f(u / \xi _i ) \over (1 / \xi _i )^{1 / 2} f(1 / \xi _i )} e^{- u} \mathrm{d}u \to \int_{0}^{\infty } u^{1 / 2 - b - 1} e^{- u} \mathrm{d}u = \Ga (1 / 2 - b) \text{.} \non 
\end{align}
Since 
\begin{align}
{1 \over {\xi _i}^2} \pi \Big( {1 \over \xi _i} \Big) f \Big( {1 \over \xi _i} \Big) \sim \Big( {1 \over \xi _i} \Big) ^{1 / 2 - b} {(1 / \xi _i )^{1 / \xi _i - 1 / 2} e^{- 1 / \xi _i} \over \Ga (1 / \xi _i )} \to \infty \text{,} \non 
\end{align}
it follows that 
\begin{align}
p( \la _i ) &\sim \int_{\xi _i}^{\infty } {1 \over \xi _i} {u \over \xi _i} \pi \Big( {u \over \xi _i} \Big) f \Big( {u \over \xi _i} \Big) e^{- u} \mathrm{d}u \sim \Ga (1 / 2 - b) {1 \over {\xi _i}^2} \pi \Big( {1 \over \xi _i} \Big) f \Big( {1 \over \xi _i} \Big) \to \infty \text{.} \non 
\end{align}
Next, suppose that $b > 1 / 2$. 
Then 
\begin{align}
p( \la _i ) &\sim \int_{0}^{\infty } u \pi (u) f(u) e^{- u \xi _i} \mathrm{d}u \non \\
&\to \int_{0}^{\infty } u \pi (u) f(u) \mathrm{d}u = \int_{0}^{\infty } u^{1 / 2} \pi (u) {u^{u - 1 / 2} e^{- u} \over \Ga (u)} \mathrm{d}u \le {1 \over (2 \pi )^{1 / 2}} \int_{0}^{\infty } u^{1 / 2} \pi (u) \mathrm{d}u < \infty \text{.} \non 
\end{align}
by the monotone convergence theorem. 
Finally, suppose that $b = 1 / 2$. 
Then 
\begin{align}
p( \la _i ) &\sim \int_{0}^{\infty } u \pi (u) f(u) e^{- u \xi _i} \mathrm{d}u \non \\
&\to \int_{0}^{\infty } u \pi (u) f(u) \mathrm{d}u = \int_{0}^{\infty } u^{1 / 2} \pi (u) {u^{u - 1 / 2} e^{- u} \over \Ga (u)} \mathrm{d}u \geqslant {1 \over (2 \pi )^{1 / 2}} \int_{0}^{\infty } u^{1 / 2} \pi (u) e^{- 1 / (12 u)} \mathrm{d}u = \infty \text{.} \non 
\end{align}
by the monotone convergence theorem. 

For parts (iii) and (iv), 
\begin{align}
p( \la _i ) &\sim %
\int_{0}^{1} u \pi (u) f(u) e^{- u \xi _i} \mathrm{d}u + \int_{1}^{\infty } u \pi (u) f(u) e^{- u \xi _i} \mathrm{d}u \text{,} \non %
\end{align}
where 
\begin{align}
\int_{0}^{1} u \pi (u) f(u) e^{- u \xi _i} \mathrm{d}u \to \int_{0}^{1} u \pi (u) f(u) \mathrm{d}u < \infty \non 
\end{align}
as in part (ii). 
By integration by parts, 
\begin{align}
&\int_{1}^{\infty } u \pi (u) f(u) e^{- u \xi _i} \mathrm{d}u = \int_{1}^{\infty } {1 \over u} u^2 \pi (u) f(u) e^{- u \xi _i} \mathrm{d}u \non \\
&= \Big[ ( \log u) u^2 \pi (u) f(u) e^{- u \xi _i} \Big] _{1}^{\infty } - \int_{1}^{\infty } ( \log u) u^2 \pi (u) f(u) e^{- u \xi _i} \Big\{ {2 \over u} + {{\pi }' (u) \over \pi (u)} + {f' (u) \over f(u)} - \xi _i \Big\} \mathrm{d}u \non \\
&= - \int_{1}^{\infty } ( \log u) u^2 \pi (u) f(u) e^{- u \xi _i} \Big\{ {2 \over u} + {{\pi }' (u) \over \pi (u)} + {f' (u) \over f(u)} \Big\} \mathrm{d}u + \xi _i \int_{1}^{\infty } ( \log u) u^2 \pi (u) f(u) e^{- u \xi _i} \mathrm{d}u \text{.} \non 
\end{align}
Since 
\begin{align}
{2 \over u} + {{\pi }' (u) \over \pi (u)} + {f' (u) \over f(u)} %
&= {a + 1 / 2 \over u (1 + u)} + \log u - \psi (u) - {1 / 2 \over u} \non \\
&= {a + 1 / 2 \over u (1 + u)} + 2 \int_{0}^{\infty } {1 \over t^2 + u^2} {t \over e^{2 \pi t} - 1} \mathrm{d}t \geqslant 0 \non 
\end{align}
for all $u \in (0, \infty )$ for part (iii) and since 
\begin{align}
{2 \over u} + {{\pi }' (u) \over \pi (u)} + {f' (u) \over f(u)} %
&= {1 / 2 \over u (1 + u)} + {1 / 2 \over u (1 + u)} \Big\{ {1 \over \log (1 + 1 / u)} - u \Big\} \non \\
&\quad + {1 / 2 + a \over 1 + \log (1 + 1 / u)} {1 \over u (1 + u)} + 2 \int_{0}^{\infty } {1 \over t^2 + u^2} {t \over e^{2 \pi t} - 1} \mathrm{d}t \geqslant 0 \non 
\end{align}
for all $u \in (0, \infty )$ for part (iv), we have, by the monotone convergence theorem, 
\begin{align}
&\int_{1}^{\infty } ( \log u) u^2 \pi (u) f(u) e^{- u \xi _i} \Big\{ {2 \over u} + {{\pi }' (u) \over \pi (u)} + {f' (u) \over f(u)} \Big\} \mathrm{d}u \non \\
&\to \int_{1}^{\infty } ( \log u) u^2 \pi (u) f(u) \Big\{ {a + 1 / 2 \over u (1 + u)} + 2 \int_{0}^{\infty } {1 \over t^2 + u^2} {t \over e^{2 \pi t} - 1} dt \Big\} \mathrm{d}u \non \\
&\le \Big[ \sup_{u \in (1, \infty )} \{ u^{3 / 2} \pi (u) u^{1 / 2} f(u) \} \Big] \int_{1}^{\infty } {\log u \over u (1 + u)} \Big\{ a + {1 \over 2} + 2 {u (1 + u) \over u^2} \int_{0}^{\infty } {t \over e^{2 \pi t} - 1} dt \Big\} \mathrm{d}u < \infty \non 
\end{align}
for part (iii) and 
\begin{align}
&\int_{1}^{\infty } ( \log u) u^2 \pi (u) f(u) e^{- u \xi _i} \Big\{ {2 \over u} + {{\pi }' (u) \over \pi (u)} + {f' (u) \over f(u)} \Big\} \mathrm{d}u \non \\
&\to \int_{1}^{\infty } ( \log u) u^2 \pi (u) f(u) \Big[ {1 / 2 \over u (1 + u)} + {1 / 2 \over u (1 + u)} \Big\{ {1 \over \log (1 + 1 / u)} - u \Big\} \non \\
&\quad + {1 / 2 + a \over 1 + \log (1 + 1 / u)} {1 \over u (1 + u)} + 2 \int_{0}^{\infty } {1 \over t^2 + u^2} {t \over e^{2 \pi t} - 1} dt \Big] \mathrm{d}u \non \\
&\le \Big[ \sup_{u \in (1, \infty )} \{ u^{3 / 2} \pi (u) u^{1 / 2} f(u) \} \Big] \int_{1}^{\infty } {\log u \over u (1 + u)} \Big[ {1 \over 2} + {1 \over 2} (1 + u - u) \non \\
&\quad + {1 / 2 + a \over 1 + \log (1 + 1 / u)} + 2 {u (1 + u) \over u^2} \int_{0}^{\infty } {t \over e^{2 \pi t} - 1} dt \Big] \mathrm{d}u < \infty \non 
\end{align}
for part (iv). 
Meanwhile, 
\begin{align}
&\xi _i \int_{1}^{\infty } ( \log u) u^2 \pi (u) f(u) e^{- u \xi _i} \mathrm{d}u / \Big\{ \Big( \log {1 \over \xi _i} \Big) (1 / \xi _i )^2 \pi (1 / \xi _i ) f(1 / \xi _i ) \Big\} \non \\
&= \int_{\xi _i}^{\infty } {\log (u / \xi _i) \over \log (1 / \xi _i )} u^2 {\pi (u / \xi _i ) \over \pi (1 / \xi _i )} {f(u / \xi _i ) \over f(1 / \xi _i )} e^{- u} \mathrm{d}u \text{.} \non 
\end{align}
Note that 
\begin{align}
{\log (u / \xi _i) \over \log (1 / \xi _i )} u^2 {\pi (u / \xi _i ) \over \pi (1 / \xi _i )} {f(u / \xi _i ) \over f(1 / \xi _i )} e^{- u} \to e^{- u} \non 
\end{align}
for all $u \in (0, \infty )$. 
Also, note that 
\begin{align}
0 &\le {\log (u / \xi _i) \over \log (1 / \xi _i )} u^2 {\pi (u / \xi _i ) \over \pi (1 / \xi _i )} {f(u / \xi _i ) \over f(1 / \xi _i )} e^{- u} \non \\
&\le \{ 1 + \log (1 + u) \} u^2 u^{a - 1} \Big( {1 + 1 / \xi _i \over 1 + u / \xi _i} \Big) ^{a + 1 / 2} {1 \over u^{1 / 2}} e^{\xi _i / 12} e^{- u} \non \\
&\le \{ 1 + \log (1 + u) \} u^2 u^{a - 1} \Big( 1 + {1 \over u^{a + 1 / 2}} \Big) {1 \over u^{1 / 2}} e^{1 / 12} e^{- u} \non \\
&= e^{1 / 12} \{ 1 + \log (1 + u) \} ( u^{a + 1 / 2} + 1) e^{- u} \non 
\end{align}
for all $u \in ( \xi _i , \infty )$ for part (iii) and that 
\begin{align}
0 &\le {\log (u / \xi _i) \over \log (1 / \xi _i )} u^2 {\pi (u / \xi _i ) \over \pi (1 / \xi _i )} {f(u / \xi _i ) \over f(1 / \xi _i )} e^{- u} \non \\
&\le \{ 1 + \log (1 + u) \} u^2 {1 \over u} {1 + 1 / \xi _i \over 1 + u / \xi _i} \Big\{ {\log (1 + \xi _i ) \over \log (1 + \xi _i / u)} \Big\} ^{1 / 2} \Big\{ {1 + \log (1 + \xi _i ) \over 1 + \log (1 + \xi _i / u)} \Big\} ^{1 / 2 + a} {1 \over u^{1 / 2}} e^{\xi _i / 12} e^{- u} \non \\
&\le \{ 1 + \log (1 + u) \} u^2 {1 \over u} \Big( 1 + {1 \over u} \Big) \Big\{ {\xi _i \over ( \xi _i / u) / (1 + 1 / u)} \Big\} ^{1 / 2} (1 + u)^{1 / 2 + a} {1 \over u^{1 / 2}} e^{1 / 12} e^{- u} \non \\
&= e^{1 / 12} \{ 1 + \log (1 + u) \} (1 + u)^{3 / 2 + a} \Big( 1 + {1 \over u} \Big) ^{1 / 2} e^{- u} \non 
\end{align}
for all $u \in ( \xi _i , \infty )$ for part (iv). 
Then, by the dominated convergence theorem, 
\begin{align}
\int_{\xi _i}^{\infty } {\log (u / \xi _i) \over \log (1 / \xi _i )} u^2 {\pi (u / \xi _i ) \over \pi (1 / \xi _i )} {f(u / \xi _i ) \over f(1 / \xi _i )} e^{- u} \mathrm{d}u \to \int_{0}^{\infty } e^{- u} \mathrm{d}u = 1 \text{.} \non 
\end{align}
Thus, we conclude that 
\begin{align}
p( \la _i ) \sim \Big( \log {1 \over \xi _i} \Big) (1 / \xi _i )^2 \pi (1 / \xi _i ) f(1 / \xi _i ) \to \infty \text{.} \non 
\end{align}
This completes the proof. 
\end{proof}

\section{Lemmas}
In this section, we prove four lemmas, which will be used in the next section. 
Let 
\begin{align}
\fai (u) = u \log \Big( 1 + {1 \over u} \Big) \non 
\end{align}
for $u \in (0, \infty )$. 

\begin{lem}
\label{lem:fai} 
The function $\fai ( \cdot )$ has the following properties. 
\begin{itemize}
\setlength{\leftskip}{12pt}
\item[{\rm{(i)}}]
${\varphi }' (u) = \log (1 + 1 / u) - 1 / (1 + u) > 0$ for all $u \in (0, \infty )$. 
\item[{\rm{(ii)}}]
${\varphi }'' (u) < 0$ for all $u \in (0, \infty )$. 
\item[{\rm{(iii)}}]
$\lim_{u \to 0} \varphi (u) = 0$ and $\lim_{u \to \infty } \varphi (u) = 1$. 
\item[{\rm{(iv)}}]
$y {\varphi }^{- 1} (v / y) = v / \log \{ 1 + 1 / {\varphi }^{- 1} (v / y) \} \to 0$ as $v < y \to \infty $ for all $v \in (0, \infty )$. 
\item[{\rm{(v)}}]
${\varphi }' (u) \geqslant 1 / \{ 2 (1 + u)^2 \} $ for all $u \in (0, \infty )$. 
\item[{\rm{(vi)}}]
There exists $0 < c_1 < 1$ such that
\[{\fai }' ( \fai ^{- 1} (1 / y)) / {\fai }' ( \fai ^{- 1} (v / y)) \le (1 / v) \{ {\fai }^{- 1} (v / y) / {\fai }^{- 1} (1 / y) \} / (1 - c_1 )\]
for all $v, y \in (0, \infty )$ satisfying $y > \max \{ 1, 2 v \} $. 
\item[{\rm{(vii)}}]
There exists $c_2 > 0$ such that ${\fai }^{- 1} (v / y) / {\fai }^{- 1} (1 / y) \le v^{c_2}$ for all $v, y \in (0, \infty )$ satisfying $v > 1$ and $y > 2 v$. 
\end{itemize}
\end{lem}

\begin{proof}%
Parts (i), (ii), and (iii) are trivial. 
Part (iv) follows since 
\begin{align}
u / y &= \fai ( {\fai }^{- 1} (u / y)) = {\fai }^{- 1} (u / y) \log \{ 1 + 1 / {\fai }^{- 1} (u / y) \} \text{.} \non 
\end{align}
Part (v) follows since, by part (i), 
\begin{align}
{\fai }' (u) &= \sum_{k = 1}^{\infty } {1 \over k} {1 \over (1 + u)^k} - {1 \over 1 + u} = {1 \over 2 (1 + u)^2} + \sum_{k = 3}^{\infty } {1 \over k} {1 \over (1 + u)^k} \text{.} \non 
\end{align}
For part (vi), let 
\begin{align}
c_1 &= \sup_{0 < t < 1 / 2} {1 / \{ 1 + {\varphi }^{- 1} (t) \} \over \log \{ 1 + 1 / {\varphi }^{- 1} (t) \} } < \infty \text{.} \non 
\end{align}
Then 
\begin{align}
c_1 &= \sup_{0 < u < {\varphi }^{- 1} (1 / 2)} {1 / (1 + u) \over \log (1 + 1 / u)} \non \\
&\le \sup_{0 < u < {\varphi }^{- 1} (1 / 2)} {1 / (1 + u) \over 1 / (1 + u) + (1 / 2) / (1 + u)^2} = \sup_{0 < u < {\varphi }^{- 1} (1 / 2)} {1 + u \over 1 + u + 1 / 2} < 1 \text{.} \non 
\end{align}
Therefore, 
\begin{align}
{{\varphi }' ( {\varphi }^{- 1} (1 / y)) \over {\varphi }' ( {\varphi }^{- 1} (v / y))} &= {\log \{ 1 + 1 / {\varphi }^{- 1} (1 / y) \} - 1 / \{ 1 + {\varphi }^{- 1} (1 / y) \} \over \log \{ 1 + 1 / {\varphi }^{- 1} (v / y) \} - 1 / \{ 1 + {\varphi }^{- 1} (v / y) \} } \non \\
&\le {\log \{ 1 + 1 / {\varphi }^{- 1} (1 / y) \} \over \log \{ 1 + 1 / {\varphi }^{- 1} (v / y) \} - 1 / \{ 1 + {\varphi }^{- 1} (v / y) \} } \non \\
&= {\log \{ 1 + 1 / {\varphi }^{- 1} (1 / y) \} \over \log \{ 1 + 1 / {\varphi }^{- 1} (v / y) \} } / \Big[ 1 - {1 / \{ 1 + {\varphi }^{- 1} (v / y) \} \over \log \{ 1 + 1 / {\varphi }^{- 1} (v / y) \} } \Big] \non \\
&\le {\log \{ 1 + 1 / {\varphi }^{- 1} (1 / y) \} \over \log \{ 1 + 1 / {\varphi }^{- 1} (v / y) \} } {1 \over 1 - c_1} \non = {1 \over v} {{\varphi }^{- 1} (v / y) \over {\varphi }^{- 1} (1 / y)} {1 \over 1 - c_1} \text{,} \non 
\end{align}
where the last equality follows from part (iv). 
For part (vii), let 
\begin{align}
c_2 &= \sup_{0 < u < {\varphi }^{- 1} (1 / 2)} \Big\{ 1 + {1 \over 1 + u} / \sum_{k = 2}^{\infty } {1 \over k} {1 \over (1 + u)^k} \Big\} < \infty \text{.} \non 
\end{align}
Then, for any $1 < s < v$, since $s / y < 1 / 2$, 
\begin{align}
{s / y \over {\varphi }^{- 1} (s / y)} {1 \over {\varphi }' ( {\varphi }^{- 1} (s / y))} &\le \sup_{0 < t < 1 / 2} \Big\{ {t \over {\varphi }^{- 1} (t)} {1 \over {\varphi }' ( {\varphi }^{- 1} (t))} \Big\} \le \sup_{0 < u < {\varphi }^{- 1} (1 / 2)} \Big\{ {\varphi (u) \over u} {1 \over {\varphi }' (u)} \Big\} \non \\
&= \sup_{0 < u < {\varphi }^{- 1} (1 / 2)} {\log (1 + 1 / u) \over \log (1 + 1 / u) - 1 / (1 + u)} = c_2 \text{,} \non 
\end{align}
where the equality follows from part (i). 
Thus, 
\begin{align}
{{\varphi }^{- 1} (v / y) \over {\varphi }^{- 1} (1 / y)} &= \exp \Big\{ \Big[ \log {\varphi }^{- 1} (s / y) \Big] _{s = 1}^{s = v} \Big\} = \exp \Big\{ \int_{s = 1}^{s = v} {1 \over s} {s / y \over {\varphi }^{- 1} (s / y)} {1 \over {\varphi }' ( {\varphi }^{- 1} (s / y))} \mathrm{d}s \Big\} \non \\
&\le \exp \Big( \int_{s = 1}^{s = v} {c_2 \over s} \mathrm{d}s \Big) = v^{c_2} \text{.} \non 
\end{align}
This completes the proof. 
\end{proof}

\bigskip

For $\ep \geqslant 0$ and $j = 0, 1$, let 
\begin{align}
I_{j, \ep } (y; \de ) &= \int_{0}^{\infty } \pi (u) {\Ga (u + \de + 1) \over \Ga (u) (y + u)^{\de + 1}} \Big( {u \over y + u} \Big) ^u \Big( {u \over \de + u} \Big) ^j e^{- \ep u} \mathrm{d}u \non 
\end{align}
for $y \in (0, \infty )$ for $\de \in (0, \infty )$.

\begin{lem}
\label{lem:approximation} 
Let $\de \in (0, \infty )$. 
\begin{itemize}
\setlength{\leftskip}{12pt}
\item[{\rm{(i)}}]
For any $y \in (\de + 1, \infty )$ and any $\ep > 0$, 
\begin{align}
1 &\le {I_{1, 0} (y; \de ) / I_{1, \ep } (y; \de ) \over I_{0, 0} (y; \de ) / I_{0, \ep } (y; \de )} \le {I_{1, 0} ( \de + 1; \de ) \over I_{1, \ep } ( \de + 1; \de )} \text{.} \non 
\end{align}
\item[{\rm{(ii)}}]
As $\ep \to 0$, 
\begin{align}
{I_{1, 0} ( \de + 1; \de ) \over I_{1, \ep } ( \de + 1; \de )} = 1 \text{.} \non 
\end{align}
\end{itemize}
\end{lem}

\begin{proof}%
Let 
Part (ii) is trivial. 
For part (i), we have by the covariance inequality that 
\begin{align}
{I_{1, 0} (y; \de ) / I_{1, \ep } (y; \de ) \over I_{0, 0} (y; \de ) / I_{0, \ep } (y; \de )} &= {I_{1, 0} (y; \de ) \over I_{0, 0} (y; \de )} {I_{0, \ep } (y; \de ) \over I_{0, 0} (y; \de )} / {I_{1, \ep } (y; \de ) \over I_{0, 0} (y; \de )} \geqslant 1 \text{.} \non 
\end{align}
On the other hand, 
\begin{align}
{I_{1, 0} (y; \de ) / I_{1, \ep } (y; \de ) \over I_{0, 0} (y; \de ) / I_{0, \ep } (y; \de )} &\le {I_{1, 0} (y; \de ) \over I_{1, \ep } (y; \de )} \text{.} \non 
\end{align}
By the covariance inequality, 
\begin{align}
{\pd \over \pd y} \log {I_{1, 0} (y; \de ) \over I_{1, \ep } (y; \de )} =&  - \int_{0}^{\infty } {u + \de + 1 \over y + u} \mathrm{d}\mu (u) / \int_{0}^{\infty } \mathrm{d}\mu (u) \non \\
&+ \int_{0}^{\infty } {u + \de + 1 \over y + u} e^{- \ep u} \mathrm{d}\mu (u) / \int_{0}^{\infty } e^{- \ep u} \mathrm{d}\mu (u) \le 0 \text{,} \non 
\end{align}
where 
\begin{align}
\mathrm{d}\mu (u) &= \pi (u) {\Ga (u + \de + 1) \over \Ga (u) (y + u)^{\de + 1}} \Big( {u \over y + u} \Big) ^u {u \over \de + u} \mathrm{d}u \non 
\end{align}
for $u \in (0, \infty )$. 
Therefore, 
\begin{align}
{I_{1, 0} (y; \de ) \over I_{1, \ep } (y; \de )} \le {I_{1, 0} ( \de + 1; \de ) \over I_{1, \ep } ( \de + 1; \de )} \text{,} \non 
\end{align}
and this completes the proof. 
\end{proof}

In the remainder of this section, we fix $\de > 0$. 
For $\ep \geqslant 0$ and $j = 0, 1$, the integral $I_{j, \ep } (y; \de )$ is denoted by $I_{j, \ep } (y)$ for $y \in (0, \infty )$ for $\de \in (0, \infty )$. 
For $\ep > 0$, let 
\begin{align}
f_{j, \ep } (v; y) &= 1(0 < v < y) {1 \over {\varphi }' ( \varphi ^{- 1} (v / y))} \pi (y \varphi ^{- 1} (v / y)) \non \\
&\quad \times {\Ga (y \varphi ^{- 1} (v / y) + \de + 1) \over \Ga (y \varphi ^{- 1} (v / y)) y^{\de + 1} \{ 1 + \varphi ^{- 1} (v / y) \} ^{\de + 1}} e^{- v} \Big\{ {y \varphi ^{- 1} (v / y) \over \de + y \varphi ^{- 1} (v / y)} \Big\} ^j e^{- \ep y \varphi ^{- 1} (v / y)} \non 
\end{align}
for $v, y \in (0, \infty )$ for $j = 0, 1$.

\begin{lem}
\label{lem:change_of_variables} 
For any $y \in (0, \infty )$, any $\ep > 0$, and any $j = 0, 1$, we have 
\begin{align}
I_{j, \ep } (y) &= \int_{0}^{\infty } f_{j, \ep } (v; y) \mathrm{d}v \text{.} \non 
\end{align}
\end{lem}

\begin{proof}%
We have 
\begin{align}
I_{j, \ep } (y) &= \int_{0}^{\infty } \pi (u) {\Ga (u + \de + 1) \over \Ga (u) (y + u)^{\de + 1}} \Big( {u \over y + u} \Big) ^u \Big( {u \over \de + u} \Big) ^j e^{- \ep u} \mathrm{d}u \non \\
&= \int_{0}^{\infty } y \pi (y u) {\Ga (y u + \de + 1) \over \Ga (y u) y^{\de + 1} (1 + u)^{\de + 1}} \Big( {u \over 1 + u} \Big) ^{y u} \Big( {y u \over \de + y u} \Big) ^j e^{- \ep y u} \mathrm{d}u \non \\
&= \int_{0}^{1} y {1 \over {\varphi }' ( \varphi ^{- 1} (t))} \pi (y \varphi ^{- 1} (t)) {\Ga (y \varphi ^{- 1} (t) + \de + 1) \over \Ga (y \varphi ^{- 1} (t)) y^{\de + 1} \{ 1 + \varphi ^{- 1} (t) \} ^{\de + 1}} e^{- y t} \Big\{ {y \varphi ^{- 1} (t) \over \de + y \varphi ^{- 1} (t)} \Big\} ^j e^{- \ep y \varphi ^{- 1} (t)} \mathrm{d}t \non \\
&= \int_{0}^{y} \Big[ {1 \over {\varphi }' ( \varphi ^{- 1} (v / y))} \pi (y \varphi ^{- 1} (v / y)) \non \\
&\quad \times {\Ga (y \varphi ^{- 1} (v / y) + \de + 1) \over \Ga (y \varphi ^{- 1} (v / y)) y^{\de + 1} \{ 1 + \varphi ^{- 1} (v / y) \} ^{\de + 1}} e^{- v} \Big\{ {y \varphi ^{- 1} (v / y) \over \de + y \varphi ^{- 1} (v / y)} \Big\} ^j e^{- \ep y \varphi ^{- 1} (v / y)} \Big] \mathrm{d}v \text{,} \non 
\end{align}
which is the desired result. 
\end{proof}

\bigskip

Let 
\begin{align}
g_j (v; y) &= {\Ga ( \de + 1) \over \de ^j} {C y^{j + \al - \de - 1} \over {\varphi }' ( \varphi ^{- 1} (1 / y))} {\{ \varphi ^{- 1} (1 / y) \} ^{j + \al } \over [1 + \log \{ 1 + (1 / y) / \varphi ^{- 1} (1 / y) \} ]^{1 + \ga }} v^{j + \al } e^{- v} \non 
\end{align}
for $v \in (0, \infty )$ and $y \in (1, \infty )$ and $j = 0, 1$. 

\begin{lem}
\label{lem:g} 
Let $j = 0, 1$ and $\ep > 0$. 
\begin{itemize}
\setlength{\leftskip}{12pt}
\item[{\rm{(i)}}]
For any $v \in (0, \infty )$, we have $f_{j, \ep } (v; y) \sim g_j (v; y)$ as $y \to \infty $. 
\item[{\rm{(ii)}}]
There exists an integrable function $h_{j, \ep } (v)$ of $v \in (0, \infty )$ such that $v^{j + \al } e^{- v} f_{j, \ep } (v; y) / g_j (v; y) \le h_{j, \ep } (v)$ for all $v \in (0, \infty )$ and $y \in (1, \infty )$. 
\end{itemize}
\end{lem}

\begin{proof}%
For part (i), 
\begin{align}
f_{j, \ep } (v; y) &\sim {\pi (y \varphi ^{- 1} (v / y)) \over {\varphi }' ( \varphi ^{- 1} (v / y))} \Ga ( \de + 1) {\varphi ^{- 1} (v / y) \over y^{\de }} e^{- v} {\{ y \varphi ^{- 1} (v / y) \} ^j \over \de ^j} \non \\
&\sim {1 \over {\varphi }' ( \varphi ^{- 1} (v / y))} {C \{ y \varphi ^{- 1} (v / y) \} ^{\al - 1} \over (1 + \log [1 + 1 / \{ y \varphi ^{- 1} (v / y) \} ])^{1 + \ga }} \Ga ( \de + 1) {\varphi ^{- 1} (v / y) \over y^{\de }} e^{- v} {\{ y \varphi ^{- 1} (v / y) \} ^j \over \de ^j} \non \\
&\sim {\Ga ( \de + 1) C y^{j + \al - \de - 1} \over \de ^j {\varphi }' ( \varphi ^{- 1} (v / y))} {\{ \varphi ^{- 1} (v / y) \} ^{j + \al } \over [1 + \log \{ 1 + (v / y) / \varphi ^{- 1} (v / y) \} ]^{1 + \ga }} e^{- v} \non 
\end{align}
as $y \to \infty $ by part (iv) of Lemma \ref{lem:fai}, condition (4) in Section 3, and the fact that $1 + \log (1 + 1 / u)$ is a slowly-varying function of $u \to 0$. 
Note that 
\begin{align}
{u {\varphi }' (u) \over \varphi (u)} \to 1 \non 
\end{align}
as $u \to 0$. 
Then 
\begin{align}
{t ( \varphi ^{- 1} )' (t) \over \varphi ^{- 1} (t)} = {\varphi ( \varphi ^{- 1} (t)) \over {\varphi }' ( \varphi ^{- 1} (t)) \varphi ^{- 1} (t)} \to 1 \non 
\end{align}
and 
\begin{align}
{t ( {\varphi }' \circ \varphi ^{- 1} )' (t) \over ( {\varphi }' \circ \varphi ^{- 1} ) (t)} = {{\varphi }^{- 1} (t) {\varphi }'' ( {\varphi }^{- 1} (t)) / {\varphi }' ( {\varphi }^{- 1} (t)) \over {\varphi }^{- 1} (t) {\varphi }' ( {\varphi }^{- 1} (t)) / \varphi ( {\varphi }^{- 1} (t))} \to 0
\end{align}
as $t \to 0$. 
Also, 
\begin{align}
\lim_{t \to 0} {t {\tilde{\fai }}' (t) \over \tilde{\fai } (t)} &= \lim_{t \to 0} {[t / \{ 1 + t / {\varphi }^{- 1} (t) \} ] (1 / {\varphi }^{- 1} (t) - [t / \{ {\varphi }^{- 1} (t) \} ^2 ] / {\varphi }' ( {\varphi }^{- 1} (t))) \over 1 + \log \{ 1 + t / {\varphi }^{- 1} (t) \} } \non \\
&= \lim_{u \to 0} {\varphi (u) / u \over 1 + \varphi (u) / u} {1 - 1 / [u {\varphi }' (u) / \varphi (u)] \over 1 + \log (1 + \varphi (u) / u)} = 0 \text{,} \non 
\end{align}
where $\tilde{\fai } (t) = 1 + \log (1 + t / {\varphi }^{- 1} (t))$ for $t \in (0, 1)$. 
Therefore, 
\begin{align}
&\varphi ^{- 1} (v / y) \sim v \varphi ^{- 1} (1 / y) \text{,} \non \\
&{\varphi }' ( \varphi ^{- 1} (v / y)) \sim {\varphi }' ( \varphi ^{- 1} (1 / y)) \text{,} \quad \text{and} \non \\
&1 + \log \{ 1 + (v / y) / \varphi ^{- 1} (v / y) \} \sim 1 + \log \{ 1 + (1 / y) / \varphi ^{- 1} (1 / y) \} \non 
\end{align}
as $y \to \infty $. 
Thus, 
\begin{align}
f_{j, \ep } (v; y) &\sim {\Ga ( \de + 1) C y^{j + \al - \de - 1} \over {\de }^j {\varphi }' ( \varphi ^{- 1} (1 / y))} {\{ v \varphi ^{- 1} (1 / y) \} ^{j + \al } \over [1 + \log \{ 1 + (1 / y) / \varphi ^{- 1} (1 / y) \} ]^{1 + \ga }} e^{- v} = g_j (v; y) \text{.} \non 
\end{align}

For part (ii), we assume that $y$ is sufficiently large and we consider the ratio 
\begin{align}
{f_{j, \ep } (v; y) \over g_j (v; y)} %
&= 1(0 < v < y) {{\varphi }' ( \varphi ^{- 1} (1 / y)) \over {\varphi }' ( \varphi ^{- 1} (v / y))} \pi (y \varphi ^{- 1} (v / y)) {\Ga (y \varphi ^{- 1} (v / y) + \de + 1) \over \Ga (y \varphi ^{- 1} (v / y)) y^{\de + 1} \{ 1 + \varphi ^{- 1} (v / y) \} ^{\de + 1}} \non \\
&\quad \times \Big\{ {y \varphi ^{- 1} (v / y) \over \de + y \varphi ^{- 1} (v / y)} \Big\} ^j e^{- \ep y \varphi ^{- 1} (v / y)} {\de ^j \over \Ga ( \de + 1)} {1 / v^{j + \al } \over C y^{j + \al - \de - 1}} {[1 + \log \{ 1 + (1 / y) / \varphi ^{- 1} (1 / y) \} ]^{1 + \ga } \over \{ \varphi ^{- 1} (1 / y) \} ^{j + \al }} \text{.} \non 
\end{align}
Since 
\begin{align}
(2 \pi )^{1 / 2} x^{x - 1 / 2} e^{- x} \le \Ga (x) \le (2 \pi )^{1 / 2} x^{x - 1 / 2} e^{- x} e^{1 / (12 x)} \non 
\end{align}
for all $x \in (0, \infty )$, 
\begin{align}
&{\Ga (y \varphi ^{- 1} (v / y) + \de + 1) \over \Ga (y \varphi ^{- 1} (v / y)) y^{\de + 1} \{ 1 + \varphi ^{- 1} (v / y) \} ^{\de + 1}} = {\Ga (y \varphi ^{- 1} (v / y) + \de + 1) \varphi ^{- 1} (v / y) \over \Ga (y \varphi ^{- 1} (v / y) + 1) y^{\de } \{ 1 + \varphi ^{- 1} (v / y) \} ^{\de + 1}} \non \\
&\le {\{ y \varphi ^{- 1} (v / y) + \de + 1 \} ^{y \varphi ^{- 1} (v / y) + \de + 1 - 1 / 2} e^{- \{ y \varphi ^{- 1} (v / y) + \de + 1 \} } \over \{ y \varphi ^{- 1} (v / y) + 1 \} ^{y \varphi ^{- 1} (v / y) + 1 - 1 / 2} e^{- \{ y \varphi ^{- 1} (v / y) + 1 \} }} {\exp (1 / [12 \{ y \varphi ^{- 1} (v / y) + \de + 1 \} ]) \over y^{\de } \{ 1 + \varphi ^{- 1} (v / y) \} ^{\de + 1} / \varphi ^{- 1} (v / y)} \non \\
&= \frac{ \displaystyle \Big\{ 1 + {\de + 1 \over y \varphi ^{- 1} (v / y)} \Big\} ^{{y \varphi ^{- 1} (v / y) \over \de + 1} ( \de + 1)} }{ \displaystyle \Big\{ 1 + {1 \over y \varphi ^{- 1} (v / y)} \Big\} ^{y \varphi ^{- 1} (v / y)} } {\{ y \varphi ^{- 1} (v / y) + \de + 1 \} ^{\de + 1 / 2} \over \{ y \varphi ^{- 1} (v / y) + 1 \} ^{1 / 2} e^{\de } y^{\de }} {\exp (1 / [12 \{ y \varphi ^{- 1} (v / y) + \de + 1 \} ]) \over \{ 1 + \varphi ^{- 1} (v / y) \} ^{\de + 1} / \varphi ^{- 1} (v / y)} \non \\
&\le {\{ y \varphi ^{- 1} (v / y) + \de + 1 \} ^{\de + 1 / 2} \over \{ y \varphi ^{- 1} (v / y) + 1 \} ^{1 / 2} y^{\de }} {\exp (1 + 1 / [12 \{ y \varphi ^{- 1} (v / y) + \de + 1 \} ]) \over \{ 1 + \varphi ^{- 1} (v / y) \} ^{\de + 1} / \varphi ^{- 1} (v / y)} \non \\
&\le {\{ y \varphi ^{- 1} (v / y) + \de + 1 \} ^{\de + 1 / 2} \over \{ y \varphi ^{- 1} (v / y) + 1 \} ^{1 / 2} y^{\de }} {\varphi ^{- 1} (v / y) \over \{ 1 + \varphi ^{- 1} (v / y) \} ^{\de + 1}} e^2 \text{,} \non 
\end{align}
where the second inequality follows from the fact that $1 \le (1 + 1 / x)^x \le e$ for all $x > 0$. 
Therefore, 
\begin{align}
{f_{j, \ep } (v; y) \over g_j (v; y)} &\le %
1(0 < v < y) {{\varphi }' ( \varphi ^{- 1} (1 / y)) \over {\varphi }' ( \varphi ^{- 1} (v / y))} \pi (y \varphi ^{- 1} (v / y)) \non \\
&\quad \times {\{ y \varphi ^{- 1} (v / y) + \de + 1 \} ^{\de + 1 / 2} \over \{ y \varphi ^{- 1} (v / y) + 1 \} ^{1 / 2}} {\varphi ^{- 1} (v / y) \over \{ 1 + \varphi ^{- 1} (v / y) \} ^{\de + 1}} e^2 \non \\
&\quad \times \Big\{ {y \varphi ^{- 1} (v / y) \over \de + y \varphi ^{- 1} (v / y)} \Big\} ^j e^{- \ep y \varphi ^{- 1} (v / y)} {{\de }^j \over \Ga ( \de + 1)} {1 / v^{j + \al } \over C y^{j + \al - 1}} {[1 + \log \{ 1 + (1 / y) / \varphi ^{- 1} (1 / y) \} ]^{1 + \ga } \over \{ \varphi ^{- 1} (1 / y) \} ^{j + \al }} \text{.} \non 
\end{align}
First, suppose that $v > y / 2$. 
Then $v / y > 1 / 2$ and $y \varphi ^{- 1} (v / y) \geqslant y \varphi ^{- 1} (1 / 2) \geqslant \de + 1 \geqslant 1$. 
Therefore, by parts (i) and (v) of Lemma \ref{lem:fai} and by condition (3) in Section 3, 
\begin{align}
{f_{j, \ep } (v; y) \over g_j (v; y)} &\le 1(0 < v < y) {{\varphi }' ( \varphi ^{- 1} (1 / y)) \over {\varphi }' ( \varphi ^{- 1} (v / y))} \pi (y \varphi ^{- 1} (v / y)) \non \\
&\quad \times 2^{\de + 1 / 2} {\{ y \varphi ^{- 1} (v / y) \} ^{\de + 1 / 2} \over \{ y \varphi ^{- 1} (v / y) \} ^{1 / 2}} {\varphi ^{- 1} (v / y) \over 1 + \varphi ^{- 1} (v / y)} e^2 \non \\
&\quad \times e^{- \ep y \varphi ^{- 1} (v / y)} {\de + 1 \over \Ga ( \de + 1)} {1 / v^{j + \al } \over C y^{j + \al - 1}} {[1 + \log \{ 1 + (1 / y) / \varphi ^{- 1} (1 / y) \} ]^{1 + \ga } \over \{ \varphi ^{- 1} (1 / y) \} ^{j + \al }} \non \\
&\le M_1 1(0 < v < y) \Big[ \log \Big\{ 1 + {1 \over \varphi ^{- 1} (1 / y)} \Big\} \Big] \{ 1 + \varphi ^{- 1} (v / y) \} ^2 {1 \over y \varphi ^{- 1} (v / y)} e^{- \ep \varphi ^{- 1} (v / y) / 3} \times y \non \\
&\quad \times 2^{\de + 1 / 2} \{ y \varphi ^{- 1} (v / y) \} ^{\de } e^2 e^{- \ep y \varphi ^{- 1} (v / y) / 3} \non \\
&\quad \times e^{- \ep y \varphi ^{- 1} (1 / 2) / 3} {1 / v^{j + \al } \over y^{j + \al - 1}} {[1 + \log \{ 1 + (1 / y) / \varphi ^{- 1} (1 / y) \} ]^{1 + \ga } \over \{ \varphi ^{- 1} (1 / y) \} ^{j + \al }} \times {1 \over y} \non \\
&\le M_2 1(0 < v < y) \Big[ \log \Big\{ 1 + {1 \over \varphi ^{- 1} (1 / y)} \Big\} \Big] \non \\
&\quad \times e^{- \ep y \varphi ^{- 1} (1 / 2) / 3} {1 / v^{j + \al } \over y^{j + \al }} {[1 + \log \{ 1 + (1 / y) / \varphi ^{- 1} (1 / y) \} ]^{1 + \ga } \over \{ \varphi ^{- 1} (1 / y) \} ^{j + \al }} \non 
\end{align}
for some $M_1 , M_2 > 0$. %
Thus, 
\begin{align}
{f_{j, \ep } (v; y) \over g_j (v; y)} &\le {M_2 \over v^{j + \al }} {e^{- \ep y \varphi ^{- 1} (1 / 2) / 3} \over y^{j + \al }} {[1 + \log \{ 1 + (1 / y) / \varphi ^{- 1} (1 / y) \} ]^{1 + \ga } \over \{ \varphi ^{- 1} (1 / y) \} ^{j + \al + 1}} \non \\
&\le {M_2 \over v^{j + \al }} {[1 + \log \{ 1 + (1 / y) / \varphi ^{- 1} (1 / y) \} ]^{1 + \ga } \over e^{{\ep }' y} \{ \varphi ^{- 1} (1 / y) \} ^{j + \al + 1}} \non \\
&\le {M_2 \over v^{j + \al }} \exp \Big( - \Big[ {{\ep }' \over 1 / y} - (j + \al + 2 + \ga ) \log \Big\{1 + {1 \over \varphi ^{- 1} (1 / y)} \Big\} \Big] \Big) \le {M_2 \over v^{j + \al }} \non 
\end{align}
for ${\ep }' = \ep \varphi ^{- 1} (1 / 2) / 3$, where the third inequality follows since 
\begin{align}
{1 \over \varphi ^{- 1} (1 / y)}, 1 + \log \Big\{ 1 + {1 / y \over \varphi ^{- 1} (1 / y)} \Big\} \le 1 + {1 \over \varphi ^{- 1} (1 / y)} \non 
\end{align}
and the last inequality follows since 
\begin{align}
{{\ep }' \over 1 / y} &= {{\ep }' \over \varphi ^{- 1} (1 / y) \log \{ 1 + 1 / \varphi ^{- 1} (1 / y) \} } \geqslant (j + \al + 2 + \ga ) \log \Big\{ 1 + {1 \over \varphi ^{- 1} (1 / y)} \Big\} \non 
\end{align}
for sufficiently large $y > 0$. 
Hence, $v^{j + \al } e^{- v} f_{j, \ep } (v; y) / g_j (v; y) \le M_2 e^{- v}$. 
Next, suppose that $1 < v < y / 2$. 
Then 
\begin{align}
{f_{j, \ep } (v; y) \over g_j (v; y)} &\le 1(0 < v < y) {\pi (y \varphi ^{- 1} (v / y)) \over C \{ y \varphi ^{- 1} (v / y) \} ^{\al - 1} / (1 + \log [1 + 1 / \{ y \varphi ^{- 1} (v / y) \} ])^{1 + \ga }} \non \\
&\quad \times (C y^{\al - 1} \{ \varphi ^{- 1} (v / y) \} ^{\al - 1} / [1 + \log \{ 1 + (1 / y) / \varphi ^{- 1} (v / y) \} ]^{1 + \ga } ) \non \\
&\quad \times {{\varphi }' ( \varphi ^{- 1} (1 / y)) \over {\varphi }' ( \varphi ^{- 1} (v / y))} {\{ y \varphi ^{- 1} (v / y) + \de + 1 \} ^{\de + 1 / 2} \over \{ y \varphi ^{- 1} (v / y) + 1 \} ^{1 / 2}} {\varphi ^{- 1} (v / y) \over \{ 1 + \varphi ^{- 1} (v / y) \} ^{\de + 1}} e^2 \non \\
&\quad \times \Big\{ {y \varphi ^{- 1} (v / y) \over \de + y \varphi ^{- 1} (v / y)} \Big\} ^j e^{- \ep y \varphi ^{- 1} (v / y)} {\de ^j \over \Ga ( \de + 1)} {1 / v^{j + \al } \over C y^{j + \al - 1}} {[1 + \log \{ 1 + (1 / y) / \varphi ^{- 1} (1 / y) \} ]^{1 + \ga } \over \{ \varphi ^{- 1} (1 / y) \} ^{j + \al }} \non \\
&\le 1(0 < v < y) {\pi (y \varphi ^{- 1} (v / y)) \over C \{ y \varphi ^{- 1} (v / y) \} ^{\al - 1} / (1 + \log [1 + 1 / \{ y \varphi ^{- 1} (v / y) \} ])^{1 + \ga }} \non \\
&\quad \times (C \{ \varphi ^{- 1} (v / y) \} ^{\al - 1} / [1 + \log \{ 1 + (1 / y) / \varphi ^{- 1} (v / y) \} ]^{1 + \ga } ) \non \\
&\quad \times {{\varphi }' ( \varphi ^{- 1} (1 / y)) \over {\varphi }' ( \varphi ^{- 1} (v / y))} {\{ y \varphi ^{- 1} (v / y) + \de + 1 \} ^{\de + 1 / 2} \over \{ y \varphi ^{- 1} (v / y) + 1 \} ^{1 / 2}} {\varphi ^{- 1} (v / y) \over \{ 1 + \varphi ^{- 1} (v / y) \} ^{\de + 1}} e^2 \non \\
&\quad \times {\{ \varphi ^{- 1} (v / y) \} ^j \over \Ga ( \de + 1)} e^{- \ep y \varphi ^{- 1} (v / y)} {1 / v^{j + \al } \over C} {[1 + \log \{ 1 + (1 / y) / \varphi ^{- 1} (1 / y) \} ]^{1 + \ga } \over \{ \varphi ^{- 1} (1 / y) \} ^{j + \al }} \non \\
&= 1(0 < v < y) {\pi (y \varphi ^{- 1} (v / y)) \over C \{ y \varphi ^{- 1} (v / y) \} ^{\al - 1} / (1 + \log [1 + 1 / \{ y \varphi ^{- 1} (v / y) \} ])^{1 + \ga }} \non \\
&\quad \times \Big[ {1 + \log \{ 1 + (1 / y) / \varphi ^{- 1} (1 / y) \} \over 1 + \log \{ 1 + (1 / y) / \varphi ^{- 1} (v / y) \} } \Big] ^{1 + \ga } {{\varphi }' ( \varphi ^{- 1} (1 / y)) \over {\varphi }' ( \varphi ^{- 1} (v / y))} \Big\{ {\varphi ^{- 1} (v / y) \over \varphi ^{- 1} (1 / y)} \Big\} ^{j + \al } \non \\
&\quad \times {\{ y \varphi ^{- 1} (v / y) + \de + 1 \} ^{\de + 1 / 2} \over \{ y \varphi ^{- 1} (v / y) + 1 \} ^{1 / 2}} {1 \over \{ 1 + \varphi ^{- 1} (v / y) \} ^{\de + 1}} e^2 {1 \over \Ga ( \de + 1)} e^{- \ep y \varphi ^{- 1} (v / y)} {1 \over v^{j + \al }} \non \\
&\le 1(0 < v < y) {\pi (y \varphi ^{- 1} (v / y)) \over C \{ y \varphi ^{- 1} (v / y) \} ^{\al - 1} / (1 + \log [1 + 1 / \{ y \varphi ^{- 1} (v / y) \} ])^{1 + \ga }} \non \\
&\quad \times \Big[ {1 + \log \{ 1 + (1 / y) / \varphi ^{- 1} (1 / y) \} \over 1 + \log \{ 1 + (1 / y) / \varphi ^{- 1} (v / y) \} } \Big] ^{1 + \ga } {{\varphi }' ( \varphi ^{- 1} (1 / y)) \over {\varphi }' ( \varphi ^{- 1} (v / y))} \Big\{ {\varphi ^{- 1} (v / y) \over \varphi ^{- 1} (1 / y)} \Big\} ^{j + \al } \non \\
&\quad \times \{ y \varphi ^{- 1} (v / y) + \de + 1 \} ^{\de } {( \de + 1)^{1 / 2} e^2 \over \Ga ( \de + 1)} e^{- \ep y \varphi ^{- 1} (v / y)} {1 \over v^{j + \al }} \text{.} \non 
\end{align}
Here, by conditions (3) and (4) in Section 3, 
\begin{align}
&1(0 < v < y) {\pi (y \varphi ^{- 1} (v / y)) \over C \{ y \varphi ^{- 1} (v / y) \} ^{\al - 1} / (1 + \log [1 + 1 / \{ y \varphi ^{- 1} (v / y) \} ])^{1 + \ga }} \non \\
&\le 1(0 < v < y) \Big[ 1(y \varphi ^{- 1} (v / y) < 1) M_3 + {1(y \varphi ^{- 1} (v / y) \geqslant 1) M_4 / \{ y \varphi ^{- 1} (v / y) \} \over C \{ y \varphi ^{- 1} (v / y) \} ^{\al - 1} / (1 + \log [1 + 1 / \{ y \varphi ^{- 1} (v / y) \} ])^{1 + \ga }} \Big] \le M_5 \non 
\end{align}
for some $M_3 , M_4 , M_5 > 0$. 
Also, by parts (vi) and (vii) of Lemma \ref{lem:fai}, 
\begin{align}
&\Big[ {1 + \log \{ 1 + (1 / y) / \varphi ^{- 1} (1 / y) \} \over 1 + \log \{ 1 + (1 / y) / \varphi ^{- 1} (v / y) \} } \Big] ^{1 + \ga } {{\varphi }' ( \varphi ^{- 1} (1 / y)) \over {\varphi }' ( \varphi ^{- 1} (v / y))} \Big\{ {\varphi ^{- 1} (v / y) \over \varphi ^{- 1} (1 / y)} \Big\} ^{j + \al } \non \\
&\le \Big[ {1 + \log \{ 1 + (1 / y) / \varphi ^{- 1} (1 / y) \} \over 1 + \log \{ 1 + (1 / y) / \varphi ^{- 1} (v / y) \} } \Big] ^{1 + \ga } {1 \over v} {\varphi ^{- 1} (v / y) \over \varphi ^{- 1} (1 / y)} {1 \over 1 - c_1} \Big\{ {\varphi ^{- 1} (v / y) \over \varphi ^{- 1} (1 / y)} \Big\} ^{j + \al } \non \\
&\le \Big[ {1 + \log \{ 1 + (1 / y) / \varphi ^{- 1} (1 / y) \} \over 1 + \log \{ 1 + (1 / y) / \varphi ^{- 1} (v / y) \} } \Big] ^{1 + \ga } {1 \over v} {1 \over 1 - c_1} v^{c_2 (j + \al + 1)} \non \\
&\le \Big\{ 1 + \log {\varphi ^{- 1} (v / y) \over \varphi ^{- 1} (1 / y)} \Big\} ^{1 + \ga } {1 \over v} {1 \over 1 - c_1} v^{c_2 (j + \al + 1)} \le (1 + c_2 \log v)^{1 + \ga } {1 \over v} {1 \over 1 - c_1} v^{c_2 (j + \al + 1)} \text{,} \non 
\end{align}
where the third inequality follows since 
\begin{align}
{1 + \log \{ 1 + (1 / y) / \varphi ^{- 1} (1 / y) \} \over 1 + \log \{ 1 + (1 / y) / \varphi ^{- 1} (v / y) \} } &\le 1 + \log {1 + (1 / y) / \varphi ^{- 1} (1 / y) \over 1 + (1 / y) / \varphi ^{- 1} (v / y)} \le 1 + \log {\varphi ^{- 1} (v / y) \over \varphi ^{- 1} (1 / y)} \non 
\end{align}
by the assumption that $v > 1$. 
Furthermore, 
\begin{align}
\{ y \varphi ^{- 1} (v / y) + \de + 1 \} ^{\de } {( \de + 1)^{1 / 2} e^2 \over \Ga ( \de + 1)} e^{- \ep y \varphi ^{- 1} (v / y)} {1 \over v^{j + \al }} \le {M_6 \over v^{j + \al }} \non 
\end{align}
for some $M_6 > 0$. 
Thus, 
\begin{align}
v^{j + \al } e^{- v} {f_{j, \ep } (v; y) \over g_j (v; y)} &\le v^{j + \al } e^{- v} M_5 (1 + c_2 \log v)^{1 + \ga } {1 \over v} {1 \over 1 - c_1} v^{c_2 (j + \al + 1)} {M_6 \over v^{j + \al }} \non \\
&\le M_7 (1 + c_2 \log v)^{1 + \ga } v^{c_2 (j + \al + 1) - 1} e^{- v} \non 
\end{align}
for some $M_7 > 0$. 
Finally, suppose that $0 < v < 1$. 
Then ${\varphi }' ( \varphi ^{- 1} (1 / y)) / {\varphi }' ( \varphi ^{- 1} (v / y)) \le 1$ by part (ii) of Lemma \ref{lem:fai}. 
Also, $y \varphi ^{- 1} (v / y) \le y \varphi ^{- 1} (1 / y) \le 1$ by part (iv) of Lemma \ref{lem:fai} since $y$ is sufficiently large. 
Therefore, by condition (4) in Section 3, 
\begin{align}
{f_{j, \ep } (v; y) \over g_j (v; y)} &\le M_8 {\{ y \varphi ^{- 1} (v / y) \} ^{\al - 1} \over (1 + \log [1 + 1 / \{ y \varphi ^{- 1} (v / y) \} ])^{1 + \ga }} \non \\
&\quad \times {\{ y \varphi ^{- 1} (v / y) + \de + 1 \} ^{\de + 1 / 2} \over \{ y \varphi ^{- 1} (v / y) + 1 \} ^{1 / 2}} {\varphi ^{- 1} (v / y) \over \{ 1 + \varphi ^{- 1} (v / y) \} ^{\de + 1}} \non \\
&\quad \times {\{ \varphi ^{- 1} (v / y) \} ^j \over \{ \de + y \varphi ^{- 1} (v / y) \} ^j} e^{- \ep y \varphi ^{- 1} (v / y)} {1 / v^{j + \al } \over y^{\al - 1}} {[1 + \log \{ 1 + (1 / y) / \varphi ^{- 1} (1 / y) \} ]^{1 + \ga } \over \{ \varphi ^{- 1} (1 / y) \} ^{j + \al }} \non \\
&= M_8 {[1 + \log \{ 1 + (1 / y) / \varphi ^{- 1} (1 / y) \} ]^{1 + \ga } \over (1 + \log [1 + 1 / \{ y \varphi ^{- 1} (v / y) \} ])^{1 + \ga }} \non \\
&\quad \times {\{ y \varphi ^{- 1} (v / y) + \de + 1 \} ^{\de + 1 / 2} \over \{ y \varphi ^{- 1} (v / y) + 1 \} ^{1 / 2}} {1 \over \{ 1 + \varphi ^{- 1} (v / y) \} ^{\de + 1}} \non \\
&\quad \times {1 \over \{ \de + y \varphi ^{- 1} (v / y) \} ^j} e^{- \ep y \varphi ^{- 1} (v / y)} {1 \over v^{j + \al }} {\{ \varphi ^{- 1} (v / y) \} ^{j + \al } \over \{ \varphi ^{- 1} (1 / y) \} ^{j + \al }} \non \\
&\le M_8 ( \de + 2)^{\de + 1 / 2} {1 \over \de ^j} {1 \over v^{j + \al }} \non 
\end{align}
and hence $v^{j + \al } e^{- v} f_{j, \ep } (v; y) / g_j (v; y) \le M_8 ( \de + 2)^{\de + 1 / 2} e^{- v} / \de ^j$. 
This completes the proof. 
\end{proof}

\section{Proof of Theorem 1}
We now prove Theorem 1. 

\begin{proof}[Proof of Theorem 1]
Note that 
\begin{align}
p( u_i \mid y_i ) &\propto \pi ( u_i ) {\Ga ( u_i + \de _i + 1) \over \Ga ( u_i ) ( \de _i y_i + u_i )^{\de _i + 1}} \Big( {u_i \over \de _i y_i + u_i} \Big) ^{u_i} \text{.} \non 
\end{align}
Then 
\begin{align}
E( \ka _i \mid y_i ) &= {I_{1, 0} ( \de _i y_i ; \de _i ) \over I_{0, 0} ( \de _i y_i ; \de _i )} = {I_{1, \ep } ( \de _i y_i ; \de _i ) \over I_{0, \ep } ( \de _i y_i ; \de _i )} {I_{1, 0} ( \de _i y_i ; \de _i ) / I_{1, \ep } ( \de _i y_i ; \de _i ) \over I_{0, 0} ( \de _i y_i ; \de _i ) / I_{0, \ep } ( \de _i y_i ; \de _i )} \text{.} \non 
\end{align}
Therefore, by Lemma \ref{lem:approximation}, 
\begin{align}
E( \ka _i \mid y_i ) &\sim {I_{1, \ep } ( \de _i y_i ; \de _i ) \over I_{0, \ep } ( \de _i y_i ; \de _i )} \non 
\end{align}
uniformly in $y_i$ as $\ep \to 0$. 
Furthermore, by Lemmas \ref{lem:change_of_variables} and \ref{lem:g}, it follows from the dominated convergence theorem that for each $\ep > 0$, 
\begin{align}
{I_{1, \ep } ( \de _i y_i ; \de _i ) \over I_{0, \ep } ( \de _i y_i ; \de _i )} &= \frac{ \int_{0}^{\infty } f_{1, \ep } (v; \de _i y_i ) \mathrm{d}v }{ \int_{0}^{\infty } f_{0, \ep } (v; \de _i y_i ) \mathrm{d}v } \Big| _{\de = \de _i} \non \\
&= {\de _i y_i \over \de _i} {\fai }^{- 1} (1 / ( \de _i y_i )) \frac{ \int_{0}^{\infty } v^{1 + \al } e^{- v} f_{1, \ep } (v; \de _i y_i ) / g_{1} (v; \de _i y_i ) \mathrm{d}v }{ \int_{0}^{\infty } v^{\al } e^{- v} f_{0, \ep } (v; \de _i y_i ) / g_{0} (v; \de _i y_i ) \mathrm{d}v } \non \\
&\sim {\de _i y_i \over \de _i} {\fai }^{- 1} (1 / ( \de _i y_i )) (1 + \al ) \non 
\end{align}
as $y_i \to \infty $. 
Thus, 
\begin{align}
E( \ka _i \mid y_i ) &\sim {\de _i y_i \over \de _i} {\fai }^{- 1} (1 / ( \de _i y_i )) (1 + \al ) = o(1) \non 
\end{align}
as $y_i \to \infty $, where the equality follows from part (iv) of Lemma \ref{lem:fai}. 
\end{proof}

\begin{rem}
It follows from the above proof that we have $\ka ^{*} (y) = y \fai ^{- 1} (1 / y)$, $y \in (0, \infty )$. 
It can be seen that $\ka ^{*} (y) = 1 / \log y$ as $y \to \infty $. 
\begin{proof}
By part (iv) of Lemma \ref{lem:fai}, 
\begin{align}
{1 \over y \fai ^{- 1} (1 / y)} &= \log \Big\{ 1 + {1 \over \fai ^{- 1} (1 / y)} \Big\} \non 
\end{align}
for all $y > 1$. 
Therefore, by parts (iii) and (iv) of Lemma \ref{lem:fai}, 
\begin{align}
{1 \over y \fai ^{- 1} (1 / y)} \sim {1 \over y \fai ^{- 1} (1 / y)} - \log {1 \over y \fai ^{- 1} (1 / y)} &= \log y + \log \{ 1 + \fai ^{- 1} (1 / y) \} \sim \log y \non 
\end{align}
as $y \to \infty $. 
\end{proof}
\end{rem}

\section{The behavior of $E( \ka _i \mid y_i )$ as $y_i \to 0$}
We can also consider the behavior of $E( \ka _i \mid y_i ) = I_{1, 0} ( \de _i y_i ; \de _i ) / I_{0, 0} ( \de _i y_i ; \de _i )$ 
as $y_i \to 0$. 
Results for the SB and IRB priors are summarized in the following proposition. 

\begin{prp}
\label{prp:0y} 
Let $\de > 0$ and let $\kah (y) = I_{1, 0} (y; \de ) / I_{0, 0} (y; \de )$ 
for $y \in (0, \infty )$. 
\begin{itemize}
\item[{\rm{(i)}}]
Suppose that $\pi (u) = \pi _{\rm{SB}} (u)$ for $u \in (0, \infty )$. 
Then, as $y \to 0$, 
\begin{align}
\kah (y) \to \begin{cases} C_2 \in (0, \infty ) \text{,} & \text{if $\de < a$} \text{,} \\ 0 \text{,} & \text{if $\de \ge a$} \text{,} \end{cases} \non 
\end{align}
for some constant $C_2$. 
\item[{\rm{(ii)}}]
Suppose that $\pi (u) = \pi _{\rm{IRB}} (u)$ for $u \in (0, \infty )$. 
Then, as $y \to 0$, 
\begin{align}
\kah (y) \to 0 \text{.} \non 
\end{align}
\end{itemize}
\end{prp}

\begin{proof}
For part (i), we have by definition that 
\begin{align}
\kah (y) &= %
\int_{0}^{\infty } {u^a \over (1 + u)^{a + b}} {\Ga (u + \de + 1) \over \Ga (u + 1) (y + u)^{\de + 1}} \Big( {u \over y + u} \Big) ^u {u \over \de + u} \mathrm{d}u \non \\
&\quad / \int_{0}^{\infty } {u^a \over (1 + u)^{a + b}} {\Ga (u + \de + 1) \over \Ga (u + 1) (y + u)^{\de + 1}} \Big( {u \over y + u} \Big) ^u \mathrm{d}u \text{.} \non 
\end{align}
for $y \in (0, \infty )$. 
First, suppose that $\de < a + 1$. 
Then the result follows from the monotone convergence theorem. 
Next, suppose that $\de \ge a + 1$. 
Then 
\begin{align}
\kah (y) &= \int_{0}^{\infty } u^a \Big( {1 + y \over 1 + y u} \Big) ^{a + b} {\Ga (y u + \de + 1) \over \Ga (y u + 1) (1 + u)^{\de + 1}} \Big( {u \over 1 + u} \Big) ^{y u} {y u \over \de + y u} \mathrm{d}u \non \\
&\quad / \int_{0}^{\infty } u^a \Big( {1 + y \over 1 + y u} \Big) ^{a + b} {\Ga (y u + \de + 1) \over \Ga (y u + 1) (1 + u)^{\de + 1}} \Big( {u \over 1 + u} \Big) ^{y u} \mathrm{d}u \non 
\end{align}
for $y \in (0, \infty )$. 
Note that there exist $M_{1, 1} , \dots , M_{1, 4} > 0$ such that for all $j = 0, 1$, $y \in (0, 1)$ and $u \in (0, 1)$, 
\begin{align}
&u^a \Big( {1 + y \over 1 + y u} \Big) ^{a + b} {\Ga (y u + \de + 1) \over \Ga (y u + 1) (1 + u)^{\de + 1}} \Big( {u \over 1 + u} \Big) ^{y u} \Big( {y u \over \de + y u} \Big) ^j \non \\
&\le 1(u \le 1) M_{1, 1} u^a \Big( 1 + {1 \over 1 + u} \Big) ^{a + b} {1 \over (1 + u)^{\de + 1}} \non \\
&\quad + 1(u > 1) u^a \Big( {1 + y \over 1 + y u} \Big) ^{a + b} {e^{1 / \{ 12 (y u + \de + 1) \} } (y u + \de + 1)^{y u + \de + 1 - 1 / 2} / e^{y u + \de + 1} \over (y u + 1)^{y u + 1 - 1 / 2} / e^{y u + 1}} {1 \over (1 + u)^{\de + 1}} \text{,} \non 
\end{align}
where the second term on the right side is less than or equal to 
\begin{align}
&1(u > 1) u^a \Big( {1 + y \over 1 + y u} \Big) ^{a + b} {(y u + \de + 1)^{\de } \over (1 + u)^{\de + 1}} \Big( {y u + \de + 1 \over y u + 1} \Big) ^{y u + 1 / 2} \non \\
&\le M_{1, 2} 1(u > 1) {u^a (y u + \de + 1)^{\de } \over (1 + y u)^{a + b} (1 + u)^{\de + 1}} \Big( 1 + {\de \over y u + 1} \Big) ^{y u + 1} \non \\
&\le \begin{cases} \displaystyle M_{1, 3} 1(u > 1) {u^a \over (1 + u)^{\de + 1}} \text{,} & \text{if $\de \le a + b$} \text{,} \\ \displaystyle M_{1, 4} 1(u > 1) {u^a (u + \de + 1)^{\de - a - b} \over (1 + u)^{\de + 1}} \text{,} & \text{if $\de > a + b$} \text{.} \end{cases} \non 
\end{align}
Then, by the dominated convergence theorem, 
\begin{align}
\kah (y) &\to \int_{0}^{\infty } 0 \mathrm{d}u / \int_{0}^{\infty } {u^a \over (1 + u)^{\de + 1}} \mathrm{d}u = 0 \non 
\end{align}
as $y \to 0$. 

For part (ii), we have 
\begin{align}
\kah (y) &= %
\int_{0}^{\infty } {1 \over 1 + u} {\{ \log (1 + 1 / u) \} ^{b - 1} \over \{ 1 + \log (1 + 1 / u) \} ^{b + a}} {\Ga (u + \de + 1) \over \Ga (u + 1) (y + u)^{\de + 1}} \Big( {u \over y + u} \Big) ^u {u \over \de + u} \mathrm{d}u \non \\
&\quad / \int_{0}^{\infty } {1 \over 1 + u} {\{ \log (1 + 1 / u) \} ^{b - 1} \over \{ 1 + \log (1 + 1 / u) \} ^{b + a}} {\Ga (u + \de + 1) \over \Ga (u + 1) (y + u)^{\de + 1}} \Big( {u \over y + u} \Big) ^u \mathrm{d}u \text{.} \non 
\end{align}
Therefore, if $\de \le 1$, the result follows from the monotone convergence theorem. 
Now, suppose that $\de > 1$. 
Then 
\begin{align}
\kah (y) &= \frac{ \displaystyle \int_{0}^{\infty } {y u \pi _{IRB} (y u) \over y \pi _{IRB} (y)} {\Ga (y u + \de + 1) \over \Ga (y u + 1) (1 + u)^{\de + 1}} \Big( {u \over 1 + u} \Big) ^{y u} {y u \over \de + y u} du }{ \displaystyle \int_{0}^{\infty } {y u \pi _{IRB} (y u) \over y \pi _{IRB} (y)} {\Ga (y u + \de + 1) \over \Ga (y u + 1) (1 + u)^{\de + 1}} \Big( {u \over 1 + u} \Big) ^{y u} du } \text{.} \label{p0yp1} 
\end{align}
Note that we have for any $0 < \ep < \min \{ 1, b \} $ that 
\begin{align}
&{y u \pi _{IRB} (y u) \over y \pi _{IRB} (y)} {\Ga (y u + \de + 1) \over \Ga (y u + 1) (1 + u)^{\de + 1}} \Big( {u \over 1 + u} \Big) ^{y u} \Big( {y u \over \de + y u} \Big) ^j \non \\
&\le 2 \Big[ {\log \{ 1 + 1 / (y u) \} \over \log (1 + 1 / y)} \Big] ^{b - 1} \Big[ {1 + \log (1 + 1 / y) \over 1 + \log \{ 1 + 1 / (y u) \} } \Big] ^{b + a} {\Ga (y u + \de + 1) \over \Ga (y u + 2) (1 + u)^{\de + 1}} \non \\
&\le 1(u \le 1) 2 \Big[ {\log \{ 1 + 1 / (y u) \} \over \log (1 + 1 / y)} \Big] ^b \sup_{u' \in (0, 1)} {\Ga ( u' + \de + 1) \over \Ga ( u' + 2)} \non \\
&\quad + 1(u > 1) 2 \Big[ {\log (1 + 1 / y) \over \log \{ 1 + 1 / (y u) \} } \Big] ^{1 - \ep } {\{ 1 + \log (1 + u) \} ^{b + a} \over (1 + u)^{\de + 1}} {\Ga (y u + \de + 1) \over \Ga (y u + 2)} \non \\
&\le 1(u \le 1) 2 \Big\{ 1 + {\log (1 + 1 / u) \over \log (1 + 1 / y)} \Big\} ^b \sup_{u' \in (0, 1)} {\Ga ( u' + \de + 1) \over \Ga ( u' + 2)} \non \\
&\quad + 1(u > 1) 2 u^{1 - \ep } {\{ 1 + \log (1 + u) \} ^{b + a} \over (1 + u)^{\de + 1}} {e^{1 / \{ 12 (y u + \de + 1) \} } (y u + \de + 1)^{y u + \de + 1 - 1 / 2} / e^{y u + \de + 1} \over (y u + 2)^{y u + 2 - 1 / 2} / e^{y u + 2}} \non 
\end{align}
for all $u \in (0, 1)$ and $y \in (0, 1)$ and all $j = 0, 1$. 
Then it can be seen that there exists an integrable function of $u \in (0, \infty )$ such that it does not depend on $y$ and is greater than the integrands of $(\ref{p0yp1})$. 
Thus, by the dominated convergence theorem, $\kah (y) \to 0$ as $y \to 0$. 
This completes the proof. 
\end{proof}

\section{Posterior propriety under improper priors}
We consider the propriety of the posterior 
\begin{align}
p( \be , \ta | y) \propto \pi ( \be ) \pi ( \ta ) \prod_{i = 1}^{n} \int_{0}^{\infty } \pi ( u_i ) f_i ( \be , \ta , u_i ) d{u_i} \text{,} \non 
\end{align}
where 
\begin{align}
f_i ( \be , \ta , u_i ) &= {( \be \ta u_i )^{\ta u_i + 1} \over ( \be \ta u_i + \de _i y_i )^{\ta u_i + \de _i + 1}} {\Ga ( \ta u_i + \de _i + 1) \over \Ga ( \ta u_i + 1)} \non 
\end{align}
for $i = 1, \dots , n$. 

\begin{prp}
\label{prp:posterior_propriety} 
\hfill
\begin{itemize}
\item[{\rm{(i)}}]
The posterior is improper if $\int_{1}^{\infty } \pi ( \ta ) d\ta = \infty $. 
\item[{\rm{(ii)}}]
The posterior is improper if $\pi ( \be ) \sim 1$ as $\be \to \infty $ and if $\lim_{\ta \to 0} \pi ( \ta ) \in (0, \infty ]$. 
\item[{\rm{(iii)}}]
The posterior is improper if $\pi ( \be ) \sim \be ^{A - 1}$ as $\be \to \infty $ for some $A > 0$ and if $\pi ( u_i ) = \pi _{\rm{IRB}} ( u_i )$ for all $u_i \in (0, \infty )$.  
\item[{\rm{(iv)}}]
The posterior is proper if $\min_{1 \le i \le n} \de _i$ is moderately large, $\pi ( \be ) = \be ^{A - 1}$ for all $\be \in (0, \infty )$ for some $0 \le A \le 1$, $\pi ( \ta )$ is a gamma distribution with a moderately large shape parameter, $\pi ( u_i )$ is the SB prior, and $n$ is sufficiently large. 
\item[{\rm{(v)}}]
The posterior is proper if $\pi ( \be ) = \be ^{- 1}$ for all $\be \in (0, \infty )$, $\pi ( \ta )$ is a gamma distribution, and either $\pi ( u_i )$ is the IRB prior and $n \ge 2$ or $\pi ( u_i )$ is the SB prior. 
\end{itemize}
\end{prp}

\begin{proof}
For all $i = 1, \dots , n$, we have 
\begin{align}
f_i ( \be , \ta , u_i ) %
&\ge {( \be \ta u_i )^{\ta u_i + 1} \over ( \be \ta u_i + \de _i y_i )^{\ta u_i + \de _i + 1}} {( \ta u_i + \de _i + 1)^{\ta u_i + \de _i + 1 - 1 / 2} / e^{\ta u_i + \de _i + 1} \over e^{1 / \{ 12 ( \ta u_i + 1) \} } ( \ta u_i + 1)^{\ta u_i + 1 - 1 / 2} / e^{\ta u_i + 1}} \non \\
&\ge {1 \over \{ 1 + ( \de _i y_i ) / ( \be \ta u_i ) \} ^{\ta u_i}} {\be \ta u_i \over ( \be \ta u_i + \de _i y_i )^{\de _i + 1}} {( \ta u_i + \de _i + 1)^{\de _i} \over e^{1 / 12 + \de _i}} \non \\
&\ge {1 \over \exp \{ \de _i y_i / \be + 1 / 12 + \de _i \} } {\be \ta u_i ( \ta u_i + \de _i + 1)^{\de _i} \over ( \be \ta u_i + \de _i y_i )^{\de _i + 1}} \text{.} \non 
\end{align}
Therefore, 
\begin{align}
&\int_{(0, \infty )^2} p( \be , \ta | y) d( \be , \ta ) \non \\
&\ge C_{1, 1} \int_{(1, \infty ) \times (0, \infty )} \pi ( \be ) \pi ( \ta ) \Big\{ \prod_{i = 1}^{n} \int_{0}^{\infty } \pi ( u_i ) {\be \ta u_i ( \ta u_i + \de _i + 1)^{\de _i} \over ( \be \ta u_i + \de _i y_i )^{\de _i + 1}} d{u_i} \Big\} d( \be , \ta ) \non \\
&\ge \begin{cases} \displaystyle C_{1, 1} \int_{(1, 2) \times (1, \infty )} \pi ( \be ) \pi ( \ta ) \Big\{ \prod_{i = 1}^{n} \int_{1}^{2} \pi ( u_i ) {\ta ( \ta + \de _i + 1)^{\de _i} \over (4 \ta + \de _i y_i )^{\de _i + 1}} d{u_i} \Big\} d( \be , \ta ) \text{,} \\ \displaystyle C_{1, 1} \int_{1}^{\infty } \pi ( \be ) \Big[ \int_{1 / (2 \be )}^{2 / \be } \pi ( \ta ) \Big\{ \prod_{i = 1}^{n} \int_{1}^{2} \pi ( u_i ) {(1 / 2) ( \de _i + 1)^{\de _i} \over (4 + \de _i y_i )^{\de _i + 1}} d{u_i} \Big\} d\ta \Big] d\be \text{,} \\ \displaystyle C_{1, 1} \int_{(1, \infty ) \times (1, 2)} \pi ( \be ) \pi ( \ta ) \Big[ \prod_{i = 1}^{n} \Big\{ \int_{1 / (2 \be )}^{2 / \be } \pi ( u_i ) d{u_i} {(1 / 2) ( \de _i + 1)^{\de _i} \over (4 + \de _i y_i )^{\de _i + 1}} \Big\} \Big] d( \be , \ta ) \text{,} \end{cases} \non 
\end{align}
where 
\begin{align}
\int_{1 / (2 \be )}^{2 / \be } \pi ( u_i ) d{u_i} &\ge {1 \over B(b, a)} \int_{1 / (2 \be )}^{2 / \be } {\be / 2 \over 1 + 2 / \be } {\{ \log (1 + \be / 2) \} ^b / \log (1 + 2 \be ) \over \{ 1 + \log (1 + 2 \be ) \} ^{b + a}} d{u_i} \non \\
&\ge {1 \over B(b, a)} {1 \over \be } {\be / 2 \over 1 + 2} {\{ \log (1 + 1 / 2) \} ^b / \log (1 + 2 \be ) \over \{ 1 + \log (1 + 2 \be ) \} ^{b + a}} \non 
\end{align}
for all $\be \in (1, \infty )$. 
This proves parts (i), (ii), and (iii). 
Parts (iv) and (v) follow from Lemma \ref{lem:propriety}. 
\end{proof}

\begin{lem}
\label{lem:propriety} 
Let $0 \le A \le 1$ and suppose that $\pi ( \be ) \sim \be ^{A - 1}$ for all $\be \in (0, \infty )$. 
Let $h_1 , h_2 , h_3 \colon (0, \infty ) \to (0, \infty )$ be nondecreasing functions. 
Suppose that 
\begin{align}
&\int_{0}^{\infty } \pi (u) \Big\{ 1 + {1 \over h_3 (u)} \Big\} du \text{,} \quad \int_{0}^{\infty } \pi ( \ta ) \Big\{ 1 + {1 \over h_2 ( \ta )} \Big\} ^n d\ta \text{,} \quad \text{and} \non \\
&\int_{1}^{\infty } \pi ( \be ) \Big( \prod_{i = 1}^{n} \Big[ {1 \over ( \be ^{\min \{ 1, 1 / \sum_{j = 1}^{n} \de _j \} } + 1)^{\de _i}} + h_2 (1 / h_1 ( \be )) + h_3 ( h_1 ( \be ) / \be ^{1 - 1 / \sum_{j = 1}^{n} \de _j} ) \Big] \Big) d\be \non 
\end{align}
are finite. 
Then the posterior is proper. 
\end{lem}

\begin{proof}
Fix $i = 1, \dots , n$. 
Then 
\begin{align}
f_i ( \be , \ta , u_i ) &\le {( \be \ta u_i )^{\ta u_i + 1} \over ( \be \ta u_i + \de _i y_i )^{\ta u_i + \de _i + 1}} {e^{1 / \{ 12 ( \ta u_i + \de _i + 1) \} } ( \ta u_i + \de _i + 1)^{\ta u_i + \de _i + 1 - 1 / 2} / e^{\ta u_i + \de _i + 1} \over ( \ta u_i + 1)^{\ta u_i + 1 - 1 / 2} / e^{\ta u_i + 1}} \non \\
&\le {( \be \ta u_i )^{\ta u_i + 1} \over ( \be \ta u_i + \de _i y_i )^{\ta u_i + \de _i + 1}} {e^{1 / \{ 12 ( \de _i + 1) \} } ( \ta u_i + \de _i + 1)^{\ta u_i + \de _i + 1 - 1 / 2} \over ( \ta u_i + 1)^{\ta u_i + 1 - 1 / 2} / e^{\de _i}} \non \\
&\le N_1 {( \be \ta u_i )^{\ta u_i + 1} ( \ta u_i + \de _i + 1)^{\de _i} \over ( \be \ta u_i + \de _i y_i )^{\ta u_i + \de _i + 1}} \non 
\end{align}
for some $N_1 > 0$. 
Therefore, %
\begin{align}
&f_i ( \be , \ta , u_i ) 1( \ta u_i \le \de _i + 1) \non \\
&\le N_2 {\be \over \be + 1} {1( \ta u_i \le \de _i + 1) \over ( \be \ta u_i + \de _i y_i )^{\de _i}} \non \\
&\le \begin{cases} \displaystyle N_2 {\be \over \be + 1} {1 \over ( \be ^{1 / \sum_{j = 1}^{n} \de _j} + \de _i y_i )^{\de _i}} \text{,} \\ \displaystyle \quad \text{if $\ta u_i \ge \be ^{1 / \sum_{j = 1}^{n} \de _j - 1}$} \text{,} \\ \displaystyle N_3 {\be \over \be + 1} \Big\{ 1( \be \le 1) + 1( \be > 1) {h_2 (1 / h_1 ( \be )) \over h_2 ( \ta )} \Big\} \text{,} \\ \displaystyle \quad \text{if $\ta u_i < \be ^{1 / \sum_{j = 1}^{n} \de _j - 1}$ and $\ta \le 1 / h_1 ( \be )$} \text{,} \\ \displaystyle N_4 {\be \over \be + 1} \Big\{ 1( \be \le 1) + 1( \be > 1) {h_3 ( h_1 ( \be ) / \be ^{1 - 1 / \sum_{j = 1}^{n} \de _j} ) \over h_3 ( u_i )} \Big\} \text{,} \\ \displaystyle \quad \text{if $\ta u_i < \be ^{1 / \sum_{j = 1}^{n} \de _j - 1}$ and $\ta > 1 / h_1 ( \be )$} \text{,} \end{cases} \non \\
&\le N_5 \Big[ 1( \be \le 1) {\be \over \be + 1} + 1( \be > 1) \Big\{ {1 \over ( \be ^{1 / \sum_{j = 1}^{n} \de _j} + \de _i y_i )^{\de _i}} + h_2 (1 / h_1 ( \be )) \non \\
&\quad + h_3 ( h_1 ( \be ) / \be ^{1 - 1 / \sum_{j = 1}^{n} \de _j} ) \Big\} \Big\{ 1 + {1 \over h_2 ( \ta )} + {1 \over h_3 ( u_i )} \Big\} \Big] \non 
\end{align}
for some $N_2 , \dots , N_5 > 0$. 
On the other hand, 
\begin{align}
f_i ( \be , \ta , u_i ) 1( \ta u_i > \de _i + 1) &\le N_1 {1( \ta u_i > \de _i + 1) \over \{ 1 + ( \de _i y_i / \be ) / ( \ta u_i ) \} ^{\ta u_i}} {( \ta u_i + \de _i + 1)^{\de _i} \over ( \be \ta u_i + \de _i y_i )^{\de _i}} \non \\
&\le N_6 {1( \ta u_i > \de _i + 1) \over \{ 1 + ( \de _i y_i / \be ) / ( \de _i + 1) \} ^{\de _i + 1}} {1 \over \be ^{\de _i}} \le N_7 {\be \over (1 + \be )^{\de _i + 1}} \non 
\end{align}
for some $N_6, N_7 > 0$. 
Thus, 
\begin{align}
f_i ( \be , \ta , u_i ) &\le N_5 \Big[ 1( \be \le 1) {\be \over \be + 1} + 1( \be > 1) \Big\{ {1 \over ( \be ^{1 / \sum_{j = 1}^{n} \de _j} + \de _i y_i )^{\de _i}} + h_2 (1 / h_1 ( \be )) \non \\
&\quad + h_3 ( h_1 ( \be ) / \be ^{1 - 1 / \sum_{j = 1}^{n} \de _j} ) \Big\} \Big\{ 1 + {1 \over h_2 ( \ta )} + {1 \over h_3 ( u_i )} \Big\} \Big] + N_7 {\be \over (1 + \be )^{\de _i + 1}} \non \\
&\le N_8 \Big( 1( \be \le 1) \be + 1( \be > 1) \Big[ {1 \over ( \be ^{\min \{ 1, 1 / \sum_{j = 1}^{n} \de _j \} } + 1)^{\de _i}} + h_2 (1 / h_1 ( \be )) \non \\
&\quad + h_3 ( h_1 ( \be ) / \be ^{1 - 1 / \sum_{j = 1}^{n} \de _j} ) \Big] \Big\{ 1 + {1 \over h_2 ( \ta )} + {1 \over h_3 ( u_i )} \Big\} \Big) \text{.} \non 
\end{align}
Since $i = 1, \dots , n$ is arbitrary, it follows that 
\begin{align}
&\int_{(0, \infty )^2} p( \be , \ta | y) d( \be , \ta ) \non \\
&\le N_9 \int_{(0, \infty )^2} \Big( \pi ( \be ) \pi ( \ta ) \Big\{ 1 + {1 \over h_2 ( \ta )} \Big\} ^n \Big[ \prod_{i = 1}^{n} \int_{0}^{\infty } \Big\{ \pi ( u_i ) \Big\{ 1 + {1 \over h_3 ( u_i )} \Big\} \Big( 1( \be \le 1) \be + 1( \be > 1) \non \\
&\quad \times \Big[ {1 \over ( \be ^{\min \{ 1, 1 / \sum_{j = 1}^{n} \de _j \} } + 1)^{\de _i}} + h_2 (1 / h_1 ( \be )) + h_3 ( h_1 ( \be ) / \be ^{1 - 1 / \sum_{j = 1}^{n} \de _j} ) \Big] \Big) \Big\} d{u_i} \Big] \Big) d( \be , \ta ) \non \\
&\le N_9 \int_{0}^{\infty } \Big( \pi ( \be ) \Big[ \prod_{i = 1}^{n} \Big( 1( \be \le 1) \be + 1( \be > 1) \non \\
&\quad \times \Big[ {1 \over ( \be ^{\min \{ 1, 1 / \sum_{j = 1}^{n} \de _j \} } + 1)^{\de _i}} + h_2 (1 / h_1 ( \be )) + h_3 ( h_1 ( \be ) / \be ^{1 - 1 / \sum_{j = 1}^{n} \de _j} ) \Big] \Big) \Big] \Big) d\be \non \\
&\quad \times \Big[ \prod_{i = 1}^{n} \int_{0}^{\infty } \pi ( u_i ) \Big\{ 1 + {1 \over h_3 ( u_i )} \Big\} d{u_i} \Big] \int_{0}^{\infty } \pi ( \ta ) \Big\{ 1 + {1 \over h_2 ( \ta )} \Big\} ^n d\ta \non \\
&\le N_{10} \int_{0}^{\infty } \Big\{ \pi ( \be ) \Big( 1( \be \le 1) \be ^n \non \\
&\quad + 1( \be > 1) \prod_{i = 1}^{n} \Big[ {1 \over ( \be ^{\min \{ 1, 1 / \sum_{j = 1}^{n} \de _j \} } + 1)^{\de _i}} + h_2 (1 / h_1 ( \be )) + h_3 ( h_1 ( \be ) / \be ^{1 - 1 / \sum_{j = 1}^{n} \de _j} ) \Big] \Big) \Big\} d\be \non \\
&\le N_{10} \int_{0}^{1} \be \pi ( \be ) d\be \non \\
&\quad + N_{10} \int_{1}^{\infty } \pi ( \be ) \Big( \prod_{i = 1}^{n} \Big[ {1 \over ( \be ^{\min \{ 1, 1 / \sum_{j = 1}^{n} \de _j \} } + 1)^{\de _i}} + h_2 (1 / h_1 ( \be )) + h_3 ( h_1 ( \be ) / \be ^{1 - 1 / \sum_{j = 1}^{n} \de _j} ) \Big] \Big) d\be \text{,} \non 
\end{align}
which is finite by assumption. 
\end{proof}

\section{Proof of Theorem 2}
Here, we prove Theorem 2. 
Let $\lat = 1 / \la $ and $\lat _0 = 1 / \la _0$.

\begin{proof}[Proof of Theorem 2]
Let $\rho (x) = x - 1 - \log x$ for $x \in (0, \infty )$. 
Then $A_{\varepsilon } ( \lambda _0 ) = \{ \lambda \in (0, \infty ) \mid \rho ( \lambda _0 / \lambda ) < \varepsilon / \de \} = \{ 1 / \lat \mid \lat \in \At _{\ep } ( \lat _0 ) \} $, where $\At _{\ep } ( \lat _0 ) = \{ \lat \in (0, \infty ) \mid \rho ( \lat / \lat _0 ) < \varepsilon / \de \} $. 
Since ${\rho }' (x) \gtreqless 0$ if and only if $x \gtreqless 1$ for any $x \in (0, \infty )$, there exist $c_{\varepsilon / \de }^{\rm{L}} \in (0, 1)$ and $c_{\varepsilon / \de }^{\rm{U}} \in (0, \infty )$ such that $\rho (1 - c_{\varepsilon / \de }^{\rm{L}} ) = \rho (1 + c_{\varepsilon / \de }^{\rm{U}} ) = \varepsilon / \de $ and $\At _{\ep } ( \lat _0 ) = ( \lat _0 - c_{\varepsilon / \de }^{\rm{L}} \lat _0 , \lat _0 + c_{\varepsilon / \de }^{\rm{U}} \lat _0 )$. 
Since $c_{\varepsilon / \de }^{\rm{U}} = \varepsilon / \de + \log (1 + c_{\varepsilon / \de }^{\rm{U}} ) > \varepsilon / \de $, we have $\int_{\At _{\ep } ( \lat _0 )} d\lat > \varepsilon \lat _0 / \de $. 

Now, suppose first that $\la _0 \neq 1$. 
Then 
\begin{align*}
\mathrm{pr}( \lambda \in A_{\varepsilon } ( \lambda_0 )) &= \mathrm{pr}( \lat \in \At _{\ep } ( \lat _0 )) \non \\
&= \int_{\At _{\ep } ( \lat _0 )} \left\{ \int_{0}^{\infty } \pi (u) {u^u \over \Gamma (u)} \lat ^u e^{- \lat u} \mathrm{d}u \right\} \mathrm{d}\lat \\
&\geqslant \int_{\At _{\ep } ( \lat _0 )} \left\{ \int_{0}^{1} \pi (u) {u^u \over \Gamma (u)} (1 - c_{\varepsilon / \de }^{\rm{L}} )^u {\lat _0}^u e^{- ( \lat _0 + c_{\varepsilon / \de }^{\rm{U}} \lat _0 ) u} \mathrm{d}u \right\} \mathrm{d}\lat \\
&\geqslant {\varepsilon \lat _0 \over \de } (1 - c_{\varepsilon / \de }^{\rm{L}} ) e^{- ( \lat _0 + c_{\varepsilon / \de }^{\rm{U}} \lat _0 )} \int_{0}^{1} \pi (u) {u^u \over \Gamma (u)} {\lat _0}^u \mathrm{d}u < \infty.
\end{align*}
Therefore, by the inequality (5) in Section 3, 
\begin{align*}
R_n &\le \varepsilon + {1 \over n} \left[ \log {\de \over \varepsilon } + \log \left\{ {e^{\lat _0 + c_{\varepsilon / \de }^{\rm{U}} \lat _0} \over \lat _0 (1 - c_{\varepsilon / \de }^{\rm{L}} )} \ \Bigg/ \  \int_{0}^{1} \pi (u) {u^u \over \Gamma (u)} {\lat _0}^u \mathrm{d}u \right\} \right] 
\end{align*}
which is $O( n^{- 1} \log n)$ when $\varepsilon = 1 / n$. 

Next, suppose that $\lambda_0 = 1$ and that $\pi (u) \sim u^{- 1 - b}$ as $u \to \infty $ for some $0 < b \le 1 / 2$. 
Then we have
\begin{align*}
\mathrm{pr}( \lambda \in A_{\varepsilon } ( \lambda_0 )) &= \mathrm{pr}( \lat \in \At _{\ep } ( \lat _0 )) = \int_{\At _{\ep } ( \lat _0 )} \left\{ \int_{0}^{\infty } \pi (u) {u^u \over \Gamma (u)} \lat ^u e^{- \lat u} \mathrm{d}u \right\} \mathrm{d}\lat \\
&= \int_{\At _{\ep } ( \lat _0 )} \left\{ \int_{0}^{\infty } \pi (u) {u^u e^{- u} \over \Gamma (u)} e^{- \rho ( \lat ) u} \mathrm{d}u \right\} \mathrm{d}\lat \\
&\geqslant \int_{\At _{\ep } ( \lat _0 )} \left\{ \int_{0}^{\infty } \pi (u) {u^u e^{- u} \over \Gamma(u)} e^{- \varepsilon u / \de } \mathrm{d}u \right\} \mathrm{d}\lat \\
&\geqslant {\varepsilon \over \de } \int_{0}^{\infty } \pi (u) {u^u e^{- u} \over \Gamma(u)} e^{- \varepsilon u / \de } \mathrm{d}u.
\end{align*}
Since the inequality 
\begin{align*}
{u^u e^{- u} \over \Gamma(u)} &\geqslant {1 \over (2 \pi )^{1 / 2}} u^{1 / 2} e^{- 1 / (12 u)} 
\end{align*}
holds for all $u \in (0, \infty )$, 
\begin{align*}
\int_{0}^{\infty } \pi (u) {u^u e^{- u} \over \Gamma(u)} e^{- \varepsilon u / \de } \mathrm{d}u &\geqslant \int_{0}^{\infty } \pi (u) {1 \over (2 \pi )^{1 / 2}} u^{1 / 2} e^{- 1 / (12 u)} e^{- \varepsilon u / \de } \mathrm{d}u \\
&\geqslant {e^{- 1 / 12} \over (2 \pi )^{1 / 2}} \int_{1}^{\infty } \pi (u) u^{1 / 2} e^{- \varepsilon u / \de } \mathrm{d}u. 
\end{align*}
Since $\pi (u) \geqslant C' u^{- 1 - b}$ for all $u \geqslant 1$ for some $C' > 0$ by assumption, 
\begin{align*}
{1 \over C'} \int_{1}^{\infty } \pi (u) u^{1 / 2} e^{- \varepsilon u / \de } \mathrm{d}u &\geqslant \int_{1}^{\infty } u^{- 1 / 2 - b} e^{- \varepsilon u / \de } \mathrm{d}u \geqslant \int_{1}^{\infty } u^{- 1} e^{- \varepsilon u / \de } \mathrm{d}u  \\
&= \left[ ( \log u) e^{- \varepsilon u / \de } \right] _{1}^{\infty } + {\varepsilon \over \de } \int_{1}^{\infty } ( \log u) e^{- \varepsilon u / \de } \mathrm{d}u = \int_{\varepsilon / \de }^{\infty } \left( \log {u \over \varepsilon / \de } \right) e^{- u} \mathrm{d}u \\
&\geqslant \int_{1}^{\infty } \left( \log {u \over \varepsilon / \de } \right) e^{- u} \mathrm{d}u = \left( \log {1 \over \varepsilon / \de } \right) \int_{1}^{\infty } {\log \{ u / ( \varepsilon / \de ) \} \over \log \{ 1 / ( \varepsilon / \de ) \} } e^{- u} \mathrm{d}u 
\end{align*}
when $\varepsilon < \de $. 
Note that $ \log \{ u / ( \varepsilon / \de ) \}  / \log \{ 1 / ( \varepsilon / \de ) \} \to 1$ as $\varepsilon \to 0$ and that if $\varepsilon < \de / 2$, then 
\begin{align*}
{\log \{ u / ( \varepsilon / \de ) \} \over \log \{ 1 / ( \varepsilon / \de ) \} } &= \exp \left( \left[ \log \log {s \over \ep / \de } \right] _{s = 1}^{s = u} \right) \\
&= \exp \left[ \int_{1}^{u} {1 \over \log \{ s / ( \varepsilon / \de ) \} } {1 \over s} \mathrm{d}s \right] \\
&\le \exp \left( \int_{1}^{u} {1 \over \log 2} {1 \over s} \mathrm{d}s \right) = u^{1 / \log 2} 
\end{align*}
for all $u > 1$. 
Then, by the dominated convergence theorem, 
\begin{align*}
\int_{1}^{\infty } {\log \{ u / ( \varepsilon / \de ) \} \over \log \{ 1 / ( \varepsilon / \de ) \} } e^{- u} \mathrm{d}u \to \int_{1}^{\infty } e^{- u} \mathrm{d}u = e^{- 1}  
\end{align*}
as $\varepsilon \to 0$. 
Thus, 
\begin{align*}
\mathrm{pr}( \lambda \in A_{\varepsilon } ( \la _0 )) &\geqslant {\varepsilon \over \de } {e^{- 1 / 12} \over (2 \pi )^{1 / 2}} C' \left( \log {\de \over \varepsilon } \right) \int_{1}^{\infty } {\log \{ u / ( \varepsilon / \de ) \} \over \log \{ 1 / ( \varepsilon / \de ) \} } e^{- u} \mathrm{d}u 
\end{align*}
and it follows from (5) that 
\begin{align*}
R_n &\le \varepsilon + {1 \over n} \left( \log {\de \over \ep } - \log \log {\de \over \varepsilon } + \log \left[ {e^{1 / 12} (2 \pi )^{1 / 2} \over C'} \ \Bigg/ \  \int_{1}^{\infty } {\log \{ u / ( \varepsilon / \de ) \} \over \log \{ 1 / ( \varepsilon / \de ) \} } e^{- u} \mathrm{d}u \right] \right),
\end{align*}
which is $O( n^{- 1} ( \log n - \log \log n))$ when $\varepsilon = 1 / n$. 
This completes the proof. 
\end{proof}

\section{Posterior sampling under the IRB prior}
By using the integral expression for the IRB density given in the following proposition, we can construct an MCMC algorithm in which $u_1 , \dots , u_n$ are easily updated. 

\begin{prp}
\label{prp:IRB} 
The IRB prior can be expressed as 
\begin{align}
&B(b, a) {\rm{IRB}} ( u_i \mid b, a) \non \\
&= \int_{(0, \infty )^3} {{s_i}^{- b} \over \Ga (1 - b)} {{w_i}^{b + a - 1} \over \Ga (b + a)} e^{- w_i} {{z_i}^{s_i + w_i} \over \Ga ( s_i + w_i + 1)} e^{- z_i} {u_i}^{s_i + w_i - 1} e^{- z_i u_i} \mathrm{d}( s_i , w_i , z_i ) \non \\
&= \int_{(0, \infty )^3} \Big[ {\rm{Ga}} ( s_i \mid 1 - b, \log (1 + 1 / u_i )) {\rm{Ga}} ( w_i \mid b + a, 1 + \log (1 + 1 / u_i )) {\rm{Ga}} ( z_i \mid s_i + w_i + 1, 1 + u_i ) \non \\
&\quad \times {1 \over u_i (1 + u_i )} {1 \over \{ \log (1 + 1 / u_i ) \} ^{1 - b}} {1 \over \{ 1 + \log (1 + 1 / u_i ) \} ^{b + a}} \Big] \mathrm{d}( s_i , w_i , z_i ) \text{.} \non 
\end{align}
\end{prp}

\begin{proof}
We have 
\begin{align}
&B(b, a) {\rm{IRB}} ( u_i \mid b, a) \non \\
&= {1 \over u_i (1 + u_i )} \int_{0}^{\infty } {{s_i}^{- b} \over \Ga (1 - b)} e^{- s_i \log (1 + 1 / u_i )} \mathrm{d}{s_i} \int_{0}^{\infty } {{w_i}^{b + a - 1} \over \Ga (b + a)} e^{- w_i} e^{- w_i \log (1 + 1 / u_i )} \mathrm{d}{w_i} \non \\
&= \int_{(0, \infty )^2} {{s_i}^{- b} \over \Ga (1 - b)} {{w_i}^{b + a - 1} \over \Ga (b + a)} e^{- w_i} {1 \over u_i (1 + u_i )} {1 \over (1 + 1 / u_i )^{s_i + w_i}} \mathrm{d}( s_i , w_i ) \non \\
&= \int_{(0, \infty )^3} {{s_i}^{- b} \over \Ga (1 - b)} {{w_i}^{b + a - 1} \over \Ga (b + a)} e^{- w_i} {{z_i}^{s_i + w_i} \over \Ga ( s_i + w_i + 1)} e^{- z_i} {u_i}^{s_i + w_i - 1} e^{- z_i u_i} \mathrm{d}( s_i , w_i , z_i ) \non \\
&= \int_{(0, \infty )^3} {{s_i}^{- b} \over \Ga (1 - b)} {{w_i}^{b + a - 1} \over \Ga (b + a)} e^{- w_i} {(1 + u_i )^{s_i + w_i + 1} {z_i}^{s_i + w_i} e^{- z_i (1 + u_i )} \over \Ga ( s_i + w_i + 1)} {{u_i}^{s_i + w_i - 1} \over (1 + u_i )^{s_i + w_i + 1}} \mathrm{d}( s_i , w_i , z_i ) \non \\
&= \int_{(0, \infty )^3} {{s_i}^{- b} \over \Ga (1 - b)} {{w_i}^{b + a - 1} \over \Ga (b + a)} e^{- w_i} {\rm{Ga}} ( z_i \mid s_i + w_i + 1, 1 + u_i ) {{u_i}^{s_i + w_i - 1} \over (1 + u_i )^{s_i + w_i + 1}} \mathrm{d}( s_i , w_i , z_i ) \non \\
&= \int_{(0, \infty )^3} \Big[ {\rm{Ga}} ( s_i \mid 1 - b, \log (1 + 1 / u_i )) {\rm{Ga}} ( w_i \mid b + a, 1 + \log (1 + 1 / u_i )) {\rm{Ga}} ( z_i \mid s_i + w_i + 1, 1 + u_i ) \non \\
&\quad \times {1 \over u_i (1 + u_i )} {1 \over \{ \log (1 + 1 / u_i ) \} ^{1 - b}} {1 \over \{ 1 + \log (1 + 1 / u_i ) \} ^{b + a}} \Big] \mathrm{d}( s_i , w_i , z_i ) \non 
\end{align}
and this proves the proposition. 
\end{proof}

We make the change of variables $\nu _i = \ta u_i$ for $i = 1, \dots , n$. 
We consider $(s, w, z) \in (0, \infty )^{3 n}$ as a set of additional latent variables. 
The overall posterior distribution of $(\lambda,\beta,\tau, s, w, z, \nu )$ given $y$ is 
\begin{align*}
&p( \lambda , \beta , \tau , s, w, z, \nu \mid y ) \non \\
&\propto \pi _{\beta} ( \beta ) \pi _{\tau } ( \tau ) {1 \over \tau ^n} \prod_{i = 1}^{n} \Big\{ {\be ^{\nu _i + 1} {\nu _i}^{\nu _i} \over \Ga ( \nu _i )} {1 \over {\lambda_i}^{\nu _i + 2}} e^{- \beta \nu _i / \la _i} {1 \over {\lambda_i}^{\de _i}} \exp \Big( - {\de _i y_i \over \lambda_i} \Big) \non \\
&\quad \times {{s_i}^{- b} \over \Ga (1 - b)} {{w_i}^{b + a - 1} \over \Ga (b + a)} e^{- w_i} {{z_i}^{s_i + w_i} \over \Ga ( s_i + w_i + 1)} e^{- z_i} {{\nu _i}^{s_i + w_i - 1} \over \ta ^{s_i + w_i - 1}} e^{- z_i \nu _i / \ta } \Big\} \non \\
&= \pi _{\beta} ( \beta ) \pi _{\tau } ( \tau ) {1 \over \tau ^n} \prod_{i = 1}^{n} \Big[ {\be ^{\nu _i + 1} {\nu _i}^{\nu _i} \over \Ga ( \nu _i )} {1 \over {\lambda_i}^{\nu _i + 2}} e^{- \beta \nu _i / \la _i} {1 \over {\lambda_i}^{\de _i}} \exp \Big( - {\de _i y_i \over \lambda_i} \Big) \non \\
&\quad \times {\rm{Ga}} ( s_i \mid 1 - b, \log (1 + \ta / \nu _i )) {\rm{Ga}} ( w_i \mid b + a, 1 + \log (1 + \ta / \nu _i )) {\rm{Ga}} ( z_i \mid s_i + w_i + 1, 1 + \nu _i / \ta ) \non \\
&\quad \times {1 \over ( \nu _i / \ta ) (1 + \nu _i / \ta )} {1 \over \{ \log (1 + \ta / \nu _i ) \} ^{1 - b}} {1 \over \{ 1 + \log (1 + \ta / \nu _i ) \} ^{b + a}} \Big] \text{.} \non 
\end{align*}

The variables $\lambda$, $\beta$, $\tau$, $(s, w, z)$, and $\nu $ are updated in the following way. 
\begin{itemize}
\setlength{\leftskip}{12pt}
\item[-]
Sample $\lambda_i \sim {\rm{IG}} ( \de _i + \nu _i + 1, \de _i y_i + \beta \nu _i )$ independently for $i = 1, \dots , n$. 
\item[-]
Sample $\beta \sim {\rm{Ga}} \big( \sum_{i=1}^n \nu _i + n + a_{\be } , \sum_{i=1}^n \nu _i / \lambda_i + b_{\be } \big) $. 
\item[-]
Sample $\tau \sim {\rm{GIG}} \big( - \sum_{i = 1}^{n} ( s_i + w_i ) + a_{\ta } , 2 b_{\ta } , 2 \sum_{i=1}^n z_i \nu _i \big) $, which has density proportional to $\tau^{- \sum_{i = 1}^{n} ( s_i + w_i ) + a_{\ta } - 1} e^{- b_{\ta } \ta - \sum_{i=1}^n z_i \nu _i / \tau }$. 
\item[-]
Independently for $i = 1, \dots , n$, 
\begin{itemize}
\setlength{\leftskip}{12pt}
\item[1.]
sample $s_i \sim {\rm{Ga}} (1 - b, \log (1 + \ta / \nu _i ))$ and $w_i \sim {\rm{Ga}} (b + a, 1 + \log (1 + \ta / \nu _i ))$ independently and 
\item[2.]
sample $z_i \sim {\rm{Ga}} ( s_i + w_i + 1, 1 + \nu _i / \ta )$. 
\end{itemize}
\item[-]
The full conditional distribution of $\nu $ is proportional to 
\begin{align*}
&\prod_{i = 1}^{n} \{ {\rm{Ga}} ( \nu _i \mid s_i + w_i , z_i / \tau ) {\rm{Ga}} (1 / \lambda_i \mid \nu _i , \beta \nu _i ) \} , 
\end{align*}
which can be accurately approximated by using the method of \cite{Miller2019} for each $i = 1, \dots , n$. 
We use the approximate full conditional distributions as proposal distributions in independent MH steps. 
\end{itemize}

\end{document}